\documentclass[aip,jmp]{revtex4-1}

\usepackage{amsmath}
\usepackage{amssymb}
\usepackage{graphicx}
\newtheorem{theorem}{Theorem}

\newtheorem{corollary}[theorem]{Corollary}

\newtheorem{definition}[theorem]{Definition}

\newtheorem{lemma}[theorem]{Lemma}

\newtheorem{remark}[theorem]{Remark}

\newenvironment{proof}[1][Proof]{\textbf{#1.} }{\ \rule{0.5em}{0.5em}}

\begin{document}

%\preprint{}

\title{Commutativity of the adiabatic elimination limit of fast oscillatory components and the instantaneous feedback limit in 
quantum feedback networks} %Title of paper


\author{John E. Gough}
  \email{jug@aber.ac.uk}
  \affiliation{Institute of Mathematics and Physics,\\
  Aberystwyth University,\\
  Aberystwyth, Wales, SY23 3BZ, United Kingdom.}

\author{Hendra I. Nurdin}%
 \email{Hendra.Nurdin@anu.edu.au}
\affiliation{Department of Information Engineering, \\
Australian National University, \\
Canberra, ACT 0200, Australia.}

\author{Sebastian Wildfeuer}
 \email{sew08@aber.ac.uk} 
 \affiliation{Institute of Mathematics and Physics,\\
  Aberystwyth University,\\
  Aberystwyth, Wales, SY23 3BZ, United Kingdom.}


\date{\today}

\begin{abstract}
We show that, for arbitrary quantum feedback networks consisting of several
quantum mechanical components connected by quantum fields, the limit of
adiabatic elimination of fast oscillator modes in the components and the
limit of instantaneous transmission along internal quantum field connections
commute. The underlying technique is to show that both limits involve a
Schur complement procedure. The result shows that the frequently used
approximations, for instance to eliminate strongly coupled optical cavities,
are mathematically consistent.
\end{abstract}


\maketitle %\maketitle must follow title, authors, abstract and \pacs

\section{Introduction}

Adiabatic elimination is a standard modeling procedure adopted when dealing
with systems that have both slow and fast variables. Here one considers the limit in which the fast variables are
effectively relaxed to instantaneous equilibrium values, which may in turn
depend on external influences, and an effective dynamics may therefore be
deduced for the slow variables. The problem becomes more involved when the
system is driven by stochastic influences. In quantum optics, fast
oscillators driven by quantum input processes may be eliminated from the
dynamics using the limiting procedure that they are strongly coupled to the
input field processes. The first rigorous account of this limit was given by Gough and van Handel\cite{GoughvanHandel} and the resulting reduced open dynamics for the slow
degrees of freedom were obtained. Extensions of this result to general
nonlinear models with a slow-fast time scale separation were given
subsequently by Bouten, Silberfarb, and van Handel \cite{BoutenSilberfarb,BvHS}.

Adiabatic approximation is frequently used to simplify the description
of a model. In this paper we aim to investigate the correctness of applying
component-wise adiabatic elimination in quantum feedback networks with Markovian components. Here several
quantum systems may be connected by passing the output noise from one
component in as input to another. In the zero time delay limit we may model
the network as a global Markovian system \cite{GoughJamesCMP09,GoughJamesIEEE09}. 
For a certain class of quantum networks and under certain conditions, we show that the instantaneous feedback limit used to
obtain a Markovian quantum feedback network is indeed compatible with the
component-wise adiabatic elimination procedure. This is the ideal situation one would
require for modeling purposes, however, the conclusion is not immediately
obvious when treating individual cases, particularly when the architecture
of the network becomes complex. We show that for both limits the form of the
coefficients of the quantum stochastic differential equation (QSDE)
describing the limit evolution can be formulated as a Schur complement of
pre-limit coefficients. Commutativity of the Schur complementation procedure
then ensures the commutativity of the adiabatic elimination and
instantaneous feedback limits.

In section \ref{sec:QO-models} we shall review the rigorous results that exist for adiabatic
elimination of oscillator components and adapt the results to deal with
multiple oscillator elimination (the proof is deferred to the Appendix). We
show commutativity of the limits for a simple cascade of components and for
components in a non-trivial feedback loop. In section \ref{sec:AE-QFN}, we present the main
features of Schur complementation which we shall need, and show that both
limits involve Schur complementation procedures. The proof of commutativity
of the limits is then established in section \ref{sec:AE-IF-commutativity}.

\vspace{5mm}
\noindent {\bf Notation.} In this paper we will use the following notation: $i$ denotes $\sqrt{-1}$,  ${\rm ker}\,X$ (or 
${\rm ker}(X)$) denotes the kernel of an operator $X$, ${\rm im}\,X$ (or 
${\rm im}(X)$) denotes the image of an operator $X$. Also, $^*$ denotes the operator adjoint. For instance, if $X=(X_1,X_2,\ldots,X_m)$ is a row vector of operators $X_1,X_2,\ldots,X_m$ on some common Hilbert space then $X^*$ is column vector given by $X^*=(X_1^*,X_2^*,\ldots,X_m^*)^T$. Also, $\mathrm{Re}\,c $ (or $\mathrm{Re}(c)$) and $\mathrm{Im}\,c$ (or $\mathrm{Im}(c)$) denote the real and imaginary parts of a complex number $c$, respectively.

\section{Models in Quantum Optics}
\label{sec:QO-models}
\subsection{Quantum Input Components}

The standard motivation for quantum stochastic evolutions in physical models
has been via traveling quantum fields interacting in a Markovian fashion
with a given quantum mechanical system \cite{Gardiner}. The fields may be described by
idealized annihilation and creation operators $b_{j}\left( t\right) $ and $%
b_{j}\left( t\right) ^{\ast }$ respectively, for $j=1,\cdots ,n$, assumed to
satisfy singular commutation relations
\begin{equation*}
\left[ b_{j}\left( t\right) ,b_{k}\left( s\right) ^{\ast }\right] =\delta
_{jk}\,\delta \left( t-s\right) .
\end{equation*}
These are sometimes referred to as quantum input processes. From these we
may define regularized operators 
\begin{equation*}
B_{j}\left( t\right) ^{\ast }=\int_{0}^{t}b_{j}\left( s\right) ^{\ast
}ds,\quad B_{k}\left( t\right) =\int_{0}^{t}b_{k}\left( s\right) ds,\quad
\Lambda _{jk}\left( t\right) =\int_{0}^{t}b_{j}\left( s\right) ^{\ast
}b_{k}\left( s\right) ds.
\end{equation*}
The older, mathematically rigorous approach is that of Hudson and
Parthasarathy which realizes the open Markov dynamics of a system with
Hilbert space $\frak{h}$ through a dilation to a \textit{unitary} evolution
on a larger space $\frak{h}\otimes \frak{F}$. Specifically $\frak{F}$ is
Bose Fock space over $\frak{K}\otimes L^{2}[0,\infty )$ where $\frak{K}=%
\mathbb{C}^{n}$ is the colour, or multiplicity, space of the quantum inputs.
The processes $B_{j}\left( \cdot \right) ,B_{k}\left( \cdot \right) ^{\ast
},\Lambda _{jk}\left( \cdot \right) $ are then realized as concrete creation, annihilation and second quantization
operators on $\frak{F}$.

We shall now work in the category of such models: each element of the category will be an open quantum system modeling a quantum device. 
A single component with intrinsic Hilbert space $\mathfrak{h}$ is modeled as
an open quantum system with external driving space $\mathfrak{F}$  -
the Bose Fock space over a one-particle space $\mathfrak{K}\otimes L^{2}(%
\mathbb{R}_{+})$. As above, $\mathfrak{K}$ is the multiplicity space of the Bose
noise field, and we shall restrict to finite multiplicity for each component ($\mathfrak{K} \equiv \mathbb{C}^{n}$ for some multiplicity $n$). 
Taking $\left\{ e_{j}\right\} _{j=1}^{n}$ to be a
fixed orthonormal basis in $\mathfrak{K}$, we realize $B_{j}(t)$ as the
annihilation operator $B(e_{j}\otimes 1_{\left[ 0,t\right] })$ on $%
\mathfrak{F}$, with $B_{j}(t)^{\ast }$ the creator. The process $\Lambda
_{jk}(t)$ is then the differential second quantization of the one-particle operator 
$\left| e_{j}\right\rangle \left\langle e_{k}\right| \otimes \tilde{1}%
_{[0,t]}$ on $\mathfrak{K}\otimes L^{2}(\mathbb{R}_{+})$ where $\tilde{1}%
_{[0,t]}$ denotes the operation of multiplication by $1_{[0,t]}$. We remark that we have the continuous
tensor product decomposition 
\begin{equation*}
\mathfrak{F} \cong \mathfrak{F}_{t]} \otimes \mathfrak{F}_{[t}
\end{equation*}
for each $t>0$, where $\mathfrak{F}_{t]}$ is the past noise space (Fock space over $\mathfrak K \otimes L^2 [0,t]$) 
and $\mathfrak{F}_{[t}$ is the future noise space (Fock space over $\mathfrak K \otimes L^2 [t, \infty )$).
A process $X(\cdot )$ on $\mathfrak {h} \otimes \mathfrak {F}$ is then said to be adapted if for each $t >0$, 
$X(t)$ acts trivially on the future factor $\mathfrak{F}_{[t}$.

The Hudson-Parthasarathy theory of quantum stochastic processes \cite{HP,partha} gives a non-commutative generalization 
of It\={o}'s stochastic integral calculus. With differentials $dB_j(t)$, $dB_k^*(t)$, $d\Lambda_{jk}(t)$ understood as being It\={o}
increment \cite{HP,partha} (i.e., they are ``forward looking", $dX(t)=X(t+dt)-X(t)$ where $X$ can be any of $B_j,B_k^*,\Lambda_{jk}$), we obtain the
following quantum It\={o} table \cite{HP,partha} for second-order products of the quantum
It\={o} differentials  (omitting the dependence on $t$ for brevity)
%\[
%\begin{tabular}{l|llll}
%$\times $ & $dB_{i}$ & $d\Lambda _{ij}$ & $dB_{j}^{\ast }$ & $dt$ \\ \hline
%$dB_{k}$ & 0 & $\delta _{ki}dB_{j}$ & $\delta _{kj}dt$ & 0 \\ 
%$d\Lambda _{kl}$ & 0 & $\delta _{li}d\Lambda _{kj}$ & $\delta
%_{lj}dB_{k}^{\ast }$ & 0 \\ 
%$dB_{l}^{\ast }$ & 0 & 0 & 0 & 0 \\ 
%$dt$ & 0 & 0 & 0 & 0
%\end{tabular}
%\]
\[
\begin{tabular}{l|llll}
$\times $ & $dB_{j}$ & $d\Lambda _{jk}$ & $dB_{k}^{\ast }$ & $dt$ \\ \hline
$dB_{l}$ & 0 & $\delta _{lj}dB_{k}$ & $\delta _{lk}dt$ & 0 \\ 
$d\Lambda _{lm}$ & 0 & $\delta _{mj}d\Lambda _{lk}$ & $\delta
_{mk}dB_{l}^{\ast }$ & 0 \\ 
$dB_{m}^{\ast }$ & 0 & 0 & 0 & 0 \\ 
$dt$ & 0 & 0 & 0 & 0
\end{tabular}
.\]

The most general form of a unitary adapted process $U(\cdot )$\ on $%
\mathfrak{h}\otimes \mathfrak{F}$, with time-independent coefficients, will
occur as the solution of a quantum stochastic differential equation (QSDE)
of the form (adopting a summation convention)
\begin{equation}
dU(t)=\left\{ 
K\otimes dt-L_{j}^{\ast }S_{jk}\otimes dB_{k}(t) 
+L_{j}\otimes dB_{j}(t)^{\ast }+(S_{jk}-\delta _{jk})\otimes d\Lambda
_{jk}(t)
\right\} U(t),  \label{HP_QSDE}
\end{equation}
with $U(0)=I$, and where the damping term is 
\begin{equation}
K=-\frac{1}{2}L_{j}^{\ast }L_{j}-iH.  \label{damping_K}
\end{equation}
We set 
\[
L=\left[ 
\begin{array}{c}
L_{1} \\ 
\vdots  \\ 
L_{n}
\end{array}
\right] ,\quad S=\left[ 
\begin{array}{ccc}
S_{11} & \cdots  & S_{1n} \\ 
\vdots  & \ddots  & \vdots  \\ 
S_{n1} & \cdots  & S_{nn}
\end{array}
\right] .
\]
We are required to take $S=\left[ S_{jk}\right] \in \mathcal{B}(\mathfrak{h}%
\otimes \mathfrak{K})$ to be unitary and $H$ to be self-adjoint. The operators
$L_j$ and $H$ are assumed to have a common dense domain in $\mathfrak{h}$, which holds in particular when these operators are bounded. In the case that they are unbounded,  they will be of a particular form which will be given in (\ref{scaled parameters}).
%For
%technical reasons we restrict to the situation where the $S_{ij},L_{i}$ and $%
%H$ are elements of $\mathcal{B}(\mathfrak{h})$, that is bounded operators.

\bigskip

\begin{figure}[h]
\centering
\includegraphics[width=0.40\textwidth]{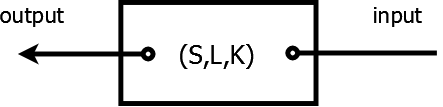}
\caption{Single component}
\label{fig:single}
\end{figure}

From our point of view the category of all possible components is parameterized by $\mathfrak h , n$ and the possible triples $(S,L,K)$ as above. 
It is convenient to collect all the coefficients in the QSDE (\ref{HP_QSDE})
into a single operator $\mathbf{G} \in \mathcal{B}( \mathfrak{H}) $ where 
\begin{equation}
\mathfrak{H}=\mathfrak{h}\otimes ( \mathbb{C}\oplus \mathfrak{K}) .
\end{equation}
With respect to the decomposition $\mathbb{C}\oplus \mathfrak{K}$ we
specifically define $\mathbf{G}$ to be 
\begin{equation}
\mathbf{G}=\left[ 
\begin{array}{cc}
K & -L^{\ast }S \\ 
L & S-I
\end{array}
\right] .
\end{equation}
In this representation, $\mathbf{G}$ appears as a  $( 1+n) $-dimensional square
matrix with entries in $\mathcal{B}( \mathfrak{h}) $.

In Fig.~\ref{fig:single} we sketch the open system as an input-output
device specified by the triple of operators $(S,L,K)$. The output fields is defined to
be the canonical processes 
\begin{equation}
B_{j}^{\mathrm{out}}(t)=U(t)^{\ast }[I\otimes B_{j}(t)]U(t).
\end{equation}

\subsection{Systems in Series}
Let us consider a pair of systems $(S_{j},L_{j},K_{j})$, $j=1,2$, connected
in series as shown in Fig.~\ref{fig:series2}. (Note that we technically do not require the observables of the two systems to commute!).

\begin{figure}[htbp]
\centering
\includegraphics[width=0.70\textwidth]{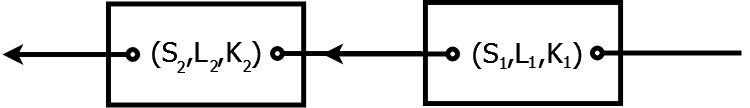}
\caption{Systems in series}
\label{fig:series2}
\end{figure}

In the instantaneous feedforward limit, the pair can be viewed as the single
system shown in Fig. \ref{fig:series} with overall parameters \cite{GoughJamesIEEE09}
\begin{equation}
(S_{\mathrm{ser}},L_{\mathrm{ser}},K_{\mathrm{ser}})=(S_{2},L_{2},K_{2})%
\vartriangleleft (S_{1},L_{1},K_{1})
\end{equation}
where the series product $\vartriangleleft$ is the associative (though non-commutative) product given by the explicit identification
\begin{eqnarray}
S_{\mathrm{ser}} &=&S_{2}S_{1}, \\
L_{\mathrm{ser}} &=&L_{2}+S_{2}L_{1}, \\
K_{\mathrm{ser}} &=&K_{1}+K_{2}-L_{2}^{\ast }S_{2}L_{1}.
\end{eqnarray}

\begin{figure}[h]
\centering
\includegraphics[width=0.40\textwidth]{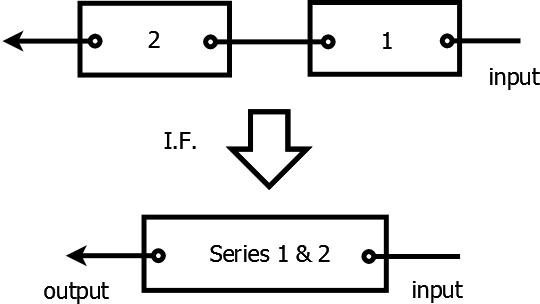}
\caption{Systems in series: the upper setup describes two systems
connected in series with a time lag $\tau>0$ in the interconnection from system 1 to 2. In the instantaneous
feedforward (I.F.) limit we consider $\protect\tau \rightarrow 0$ in which
case we obtain an effective single component model again.}
\label{fig:series}
\end{figure}

Note that if $H_{j}$ $( j=1,2) $ are the Hamiltonians of the separate
systems then the damping operators are $K_{j}=-\frac{1}{2}L_{j}^{\ast
}L_{j}-iH_{j}$ and the effective Hamiltonian in series is then given by 
\[
H_{\mathrm{ser}}=H_{1}+H_{2}+\mathrm{Im}(L_{2}^{\ast
}S_{2}L_{1}) .
\]

\subsection{Adiabatic Elimination of Oscillators}

We suppose that the system consists of local oscillators having Hilbert
space $\mathfrak{h}_{\mathrm{osc}}$\ and that the remaining degrees of
freedom live on an auxiliary space $\mathfrak{\hat{h}}$. The overall Hilbert
space of the system is then $\mathfrak{\hat{h}}\otimes \mathfrak{h}_{\mathrm{%
osc}}$ and we consider an open model described by the triple of operators 
\begin{eqnarray}
S(k) &=&S\otimes I,  \notag \\
L(k) &=&k\sum_{j}C_{j}\otimes a_{j}+G\otimes I,  \notag \\
K(k) &=&k^{2}\sum_{jl}A_{jl}\otimes a_{j}^{\ast }a_{l}+k\sum_{j}Z_{j}\otimes
a_{j}^{\ast }+k\sum_{j}X_{j}\otimes a_{j}+R\otimes I,
\label{scaled parameters}
\end{eqnarray}
where $k$ is a positive scaling parameter and $%
S,C_{j},G,A_{jl},X_{j},Z_{j},R $ are bounded operators on $\mathfrak{\hat{h}}
$\ with $A=\left[ A_{jl}\right] $ boundedly invertible. Here $a_{j}$ is the
annihilator corresponding to the $j$-th local oscillator, say with $%
j=1,\cdots ,m$.

As $k\rightarrow \infty $ the oscillators become increasingly strongly
coupled to the driving noise field and in this limit we would like to
consider them as being permanently relaxed to their joint ground state. The
oscillators are then the fast degrees of freedom of the system, with the
auxiliary space $\mathfrak{\hat{h}}$ describing the slow degrees. In the
adiabatic elimination $k\rightarrow \infty $ we desire a reduced description
of an open system involving the operators of $\mathfrak{\hat{h}}$ only, with
the fast oscillators being eliminated, as illustrated in Fig. \ref{fig:AEsingle}. The ground state for the ensemble of $%
m$ oscillators will be denoted by $\left| 0\right\rangle _{\mathrm{osc}}$.

Define $X$ to be the {\em row} vector of operators $X=(X_1,X_2,\ldots,X_m)$ and
$Z$ to be the {\em column} vector of operators $Z=(Z_1,Z_2,\ldots,Z_m)^T$. Also, we say
that a matrix  $M=[M_{jl}]_{j,l=1,\ldots,m}$, with $M_{jl}$ bounded operators on $\hat{\mathfrak{h}}$, is \emph{strictly Hurwitz stable} if ${\rm Re} \langle \phi, M \phi \rangle<0$ for all $0 \neq \phi \in \hat{\mathfrak{h}} \otimes \mathbb{C}^m$. Then we say that an open Markov quantum system with parameters of the form (\ref{scaled parameters}) is strictly Hurwitz stable if the matrix $[A_{jl}]$ is strictly Hurwitz stable. We first have the following result:

\begin{theorem}
\label{thm: adiabatic elimination} Let $U(\cdot ,k)$ be the unitary adapted
evolution associated with the triple $(S(k),L(k),K(k))$ appearing in (\ref
{scaled parameters}), and define the slow space as $\mathfrak{h}_{s}=\mathfrak{\hat{h}}%
\otimes \left\{ \mathbb{C}\left| 0\right\rangle _{\mathrm{osc}}\right\} $. If the operator $Y= \sum_{jl} A_{jl} \otimes a_j^* a_l$ has kernel space equal to the slow space, then
we have the limit 
\[
\lim_{k\rightarrow \infty }\sup_{0\leq t\leq T}\left\| U(t,k)\Phi -\hat{U}%
(t)\Phi \right\| =0,
\]
for all $T>0$ and $\Phi \in \mathfrak{h}_{s}\otimes \mathfrak{F}$, where $%
\hat{U}(\cdot )$ is the unitary evolution associated with the triple $\hat{S}%
\otimes |0\rangle \langle 0|_{\mathrm{osc}},\hat{L}\otimes |0\rangle \langle
0|_{\mathrm{osc}},\hat{K}\otimes |0\rangle \langle 0|_{\mathrm{osc}}$ and
\begin{eqnarray}
\hat{S} &=&(I+CA^{-1}C^{\ast })S,  \notag \\
\hat{L} &=&G-CA^{-1}Z,  \notag \\
\hat{K} &=&R-XA^{-1}Z.  \label{Limit parameters}
\end{eqnarray}
\end{theorem}

\begin{remark}
\label{rm:osc-notation} For ease of notation, we will drop the factor ``$\cdot \otimes |0\rangle
\langle 0|_{\mathrm{osc}}$'' as it is obvious that in the limit we are
always relaxed into the fast oscillator ground states. Therefore we can simply think of
the limit QSDE as having initial space $\mathfrak{\hat{h}}$ and coefficients $%
 ( \hat{S},\hat{L},\hat{K} ) $.
\end{remark}

\begin{remark}
A sufficient, though not necessary, condition for the kernel space of $Y$ to equal the slow space is that the matrix $A$ be strictly Hurwitz, see Lemma \ref{lem:Hurwitz}. 
\end{remark}

The result is a generalization of what has been established for the single
mode case\cite{GoughvanHandel} where the main result is stated for weak convergence of the unitaries, but this automatically extends to the strong convergence above. There the techniques
were based on a quantum central limit theorem\cite{GoughCMP} which have been shown to extend to the
multimode situation\cite{GoughJMP}.  We shall give a proof the theorem in the Appendix, exploiting the theory  of singular perturbation of QSDEs developed by Bouten, van Handel, and Silberfarb \cite{BvHS}.

In the following, we shall drop the tensor product symbol for notational
convenience. Furthermore we shall introduce the vectorial multi-mode
notation 
\[
a=\left[ 
\begin{array}{c}
a_{1} \\ 
\vdots \\ 
a_{m}
\end{array}
\right] ,a^{\ast }=\left[ a_{1}^{\ast },\cdots ,a_{m}^{\ast }\right] .
\]
We therefore write simply 
\[
S( k) =S,\;L( k) =kCa+G,\;K( k) =k^{2}a^{\ast }Aa+ka^{\ast }Z+kXa+R.
\]
If we take the Hamiltonian to be 
\[
H( k) =k^{2}a^{\ast }\Omega a+ka^{\ast }\Gamma +k\Gamma ^{\ast }a+\Theta
\]
then 
\begin{eqnarray*}
A &=&-\frac{1}{2}C^{\ast }C-i\Omega , \\
Z &=&-\frac{1}{2}C^{\ast }G-i\Gamma , \\
X &=&-\frac{1}{2}G^{\ast }C-i\Gamma ^{\ast }, \\
R &=&-\frac{1}{2}G^{\ast }G-i\Theta .
\end{eqnarray*}
In particular we note the identities 
\begin{eqnarray}
A+A^{\ast } &=&-C^{\ast }C,  \notag \\
X+Z^{\ast } &=&-G^{\ast }C,  \notag \\
R+R^{\ast } &=&-G^{\ast }G.  \label{identities}
\end{eqnarray}

\begin{figure}[tbp]
\centering
\includegraphics[width=0.70\textwidth]{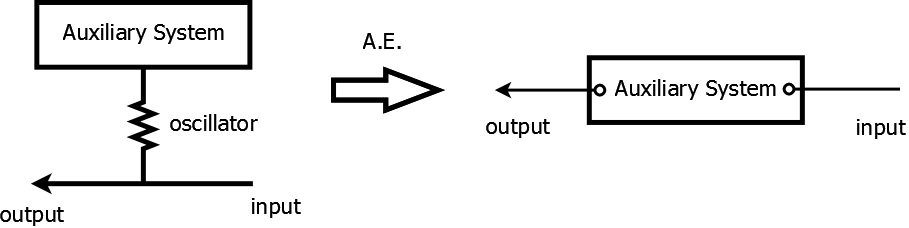}
\caption{The setup on the left shows a system of an oscillator and auxiliary
component with the oscillator coupled to a quantum field input. In the limit
where the coupling of the oscillator becomes infinitely strong, we may
adiabatically eliminate the oscillator to obtain an input acting directly on
the auxiliary component. This is sketched in the setup on the right.}
\label{fig:AEsingle}
\end{figure}

We present a na\"{i}ve derivation of the limit form appearing in Theorem \ref{thm: adiabatic elimination},
with the proof presented in the Appendix. In the interaction picture
we have the quantum Langevin equation
\begin{eqnarray*}
\dot{a} &=&\frac{1}{2}L(k)^{\ast }\left[ a,L(k)\right] +\frac{1}{2}%
L(k)^{\ast }\left[ a,L(k)\right] -i\left[ a,H(k)\right] -\left[ L(k)^{\ast
}S,a\right] b_{\mathrm{in}} \\
&=&-k^{2}(\frac{1}{2}C^{\ast }C+i\Omega )a+k(-\frac{1}{2}C^{\ast }G-i\Gamma
)-kC^{\ast }Sb_{\mathrm{in}},
\end{eqnarray*}
where $b_{\rm in}$ is an input quantum process satisfying $[b_{\rm in} (t),b_{\rm in} (s)^\ast ] = \delta (t-s)$.
Likewise the input-output relations are 
\[
b_{\mathrm{out}}=Sb_{\mathrm{in}}+L(k)=Sb_{\mathrm{in}}+(kCa+G),
\]
where $b_{\rm out}$ is the output quantum white noise field. 

We note that we may rewrite the Langevin equation as 
$\frac{1}{k}\dot{a}%
=-kAa+Z-C^{\ast }Sb_{\mathrm{in}}$ and one argues that as $k\rightarrow
\infty $ the left hand side vanishes, so that the right hand side may be
rearranged as 
\[
ka\approx A^{-1}(C^{\ast }Sb_{\mathrm{in}}-Z).
\]
The common interpretation of this is that the (scaled) oscillator mode
becomes ``slaved'' to the input field: usually this argument is given with $%
k $ fixed to unity and while clearly mathematically problematic nevertheless,
rather miraculously, yields the correct answer. Making this substitution in the output relations, we reasonably expect that 
\begin{eqnarray*}
b_{\mathrm{out}} &=&(I+CA^{-1}C^{\ast })Sb_{\mathrm{in}}+(G-CA^{-1}Z) \\
&\equiv &\hat{S}b_{\mathrm{in}}+\hat{L}.
\end{eqnarray*}
This justifies the form of $\hat{S}$ and $\hat{L}$. The form of $\hat{K}$
may be deduced by substituting $ka\approx A^{-1}(C^{\ast }Sb_{\mathrm{in}%
}-Z) $, and the conjugate relation, into the Langevin equation for any
operator acting nontrivially only on the space $\mathfrak{\hat{h}}$. (There
is a potential operator ordering ambiguity here, and the appropriate choice
is to substitute $ka$ and $ka^{\ast }$ in Wick ordered form!)

\subsection{Adiabatic Elimination and Systems in Series}
\label{sec:sys-series}
The aim of this section is to determine whether the limits of adiabatic
elimination and instantaneous feedforward do in fact commute, as illustrated
in Fig. \ref{fig:CommutingLimits}. While this is often assumed in quantum
optics models, it is certainly far from obvious. At this stage, however, we
are able to reduce the question to a direct computation.

\begin{figure}[tbp]
\centering
\includegraphics[width=0.70\textwidth]{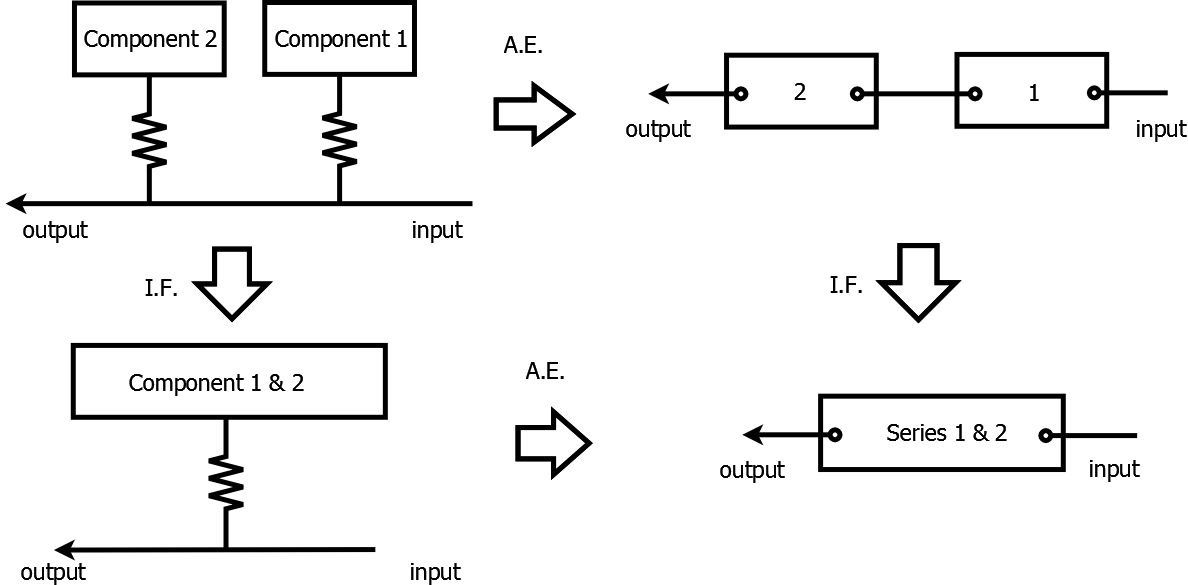}
\caption{The picture illustrates the main result that we which to prove
here, namely that the adiabatic elimination (A.E.) and the instantaneous
feedforward (I.F.) limits can be interchanged.}
\label{fig:CommutingLimits}
\end{figure}

Let us represent the local oscillators $a_{1}$ and $a_{2}$ in a combined
manner as 
\[
a=\left[ 
\begin{array}{c}
a_{1} \\ 
a_{2}
\end{array}
\right] .
\]
Then the first system is to be represented as $( S_{1}( k) ,L_{1}( k)
,K_{1}( k) ) $ where 
\begin{eqnarray*}
S_{1}( k) &=&S_{1}, \\
L_{1}( k) &=&kC_{1}a_{1}+G_{1}\equiv k[C_{1},0]a+G_{1}, \\
K_{1}( k) &=&k^{2}a_{1}^{\ast }A_{1}a_{1}+ka_{1}^{\ast }Z_{1}+kXa_{1}+R_{1}
\\
&\equiv &k^{2}a^{\ast }\left[ 
\begin{array}{cc}
A_{1} & 0 \\ 
0 & 0
\end{array}
\right] a+ka^{\ast }\left[ 
\begin{array}{c}
Z_{1} \\ 
0
\end{array}
\right] +k\left[ X_{1},0\right] a+R_{1}.
\end{eqnarray*}
Likewise, the second system is then represented as 
\begin{eqnarray*}
S_{2}( k) &=&S_{2}, \\
L_{2}( k) &\equiv &k[0,C_{2}]a+G_{2}, \\
K_{2}( k) &\equiv &k^{2}a^{\ast }\left[ 
\begin{array}{cc}
0 & 0 \\ 
0 & A_{2}
\end{array}
\right] a+ka^{\ast }\left[ 
\begin{array}{c}
0 \\ 
Z_{2}
\end{array}
\right] +k\left[ 0,X_{2}\right] a+R_{2}.
\end{eqnarray*}

\subsubsection{Adiabatic Elimination followed by Instantaneous Feedforward}

If we perform the adiabatic elimination first then we arrive at the two
systems $(j=1,2)$%
\begin{eqnarray*}
\hat{S}_{j} &=&( I+C_{j}A_{j}^{-1}C_{j}^{\ast }) S_{j} \\
\hat{L}_{j} &=&G_{j}-C_{j}A_{j}^{-1}Z_{j}, \\
\hat{K}_{j} &=&R_{j}-X_{j}A_{j}^{-1}Z_{j}.
\end{eqnarray*}
The instantaneous feedforward limit is then given by the series product

\begin{eqnarray}
S &=&\hat{S}_{2}\hat{S}_{1},  \notag \\
L &=&\hat{L}_{2}+\hat{S}_{2}\hat{L}_{1}  \notag \\
K &=&\hat{K}_{1}+\hat{K}_{2}-\hat{L}_{2}^{\ast }\hat{S}_{2}\hat{L}_{1}.
\label{AEthenIF}
\end{eqnarray}

\subsubsection{Instantaneous Feedfoward followed by Adiabatic Elimination}

We perform the series product $( S_{2}( k) ,L_{2}( k) ,K_{2}( k) )
\vartriangleleft ( S_{1}( k) ,L_{1}( k) ,K_{1}( k) ) $ first to obtain $( S_{%
\mathrm{ser}}( k) ,L_{\mathrm{ser}}( k) ,K_{\mathrm{ser}}( k) ) $ where 
\begin{eqnarray*}
S_{\mathrm{ser}}( k) &=&S_{2}S_{1}, \\
L_{\mathrm{ser}}( k) &=&L_{2}( k) +S_{2}L_{1}( k) \\
&\equiv &k[S_{2}C_{1},C_{2}]a+G_{1}+S_{2}G_{1}, \\
K_{\mathrm{ser}}( k) &=&K_{1}( k) +K_{2}( k) -L_{2}( k) ^{\ast }S_{2}( k)
L_{1}( k) \\
&=&k^{2}a^{\ast }\left[ 
\begin{array}{cc}
A_{1} & 0 \\ 
-C_{2}^{\ast }S_{2}C_{1} & A_{2}
\end{array}
\right] a+ka^{\ast }\left[ 
\begin{array}{c}
Z_{1} \\ 
Z_{2}-C_{2}^{\ast }S_{2}G_{1}
\end{array}
\right] \\
&&+k\left[ X_{1}-G_{2}^{\ast }S_{2}C_{1},X_{2}\right] a+R_{1}+R_{2}-G_{2}^{%
\ast }S_{2}G_{1}.
\end{eqnarray*}

Now, adiabatically eliminating the oscillators leads to the effective model $(%
\hat{S}_{\mathrm{ser}},\hat{L}_{\mathrm{ser}},\hat{H}_{\mathrm{ser}})$. Here
we have 
\begin{eqnarray}
\hat{S}_{\mathrm{ser}} &=&(I+[S_{2}C_{1},C_{2}]\left[ 
\begin{array}{cc}
A_{1} & 0 \\ 
-C_{2}^{\ast }S_{2}C_{1} & A_{2}
\end{array}
\right] ^{-1}\left[ 
\begin{array}{c}
C_{1}^{\ast }S_{2}^{\ast } \\ 
C_{2}^{\ast }
\end{array}
\right] )S_{2}S_{1},  \notag \\
\hat{L}_{\mathrm{ser}} &=&(G_{1}+S_{2}G_{1})-[S_{2}C_{1},C_{2}]\left[ 
\begin{array}{cc}
A_{1} & 0 \\ 
-C_{2}^{\ast }S_{2}C_{1} & A_{2}
\end{array}
\right] ^{-1}\left[ 
\begin{array}{c}
Z_{1} \\ 
Z_{2}-C_{2}^{\ast }S_{2}G_{1}
\end{array}
\right]  \notag \\
\hat{K}_{\mathrm{ser}} &=&(R_{1}+R_{2}-G_{2}^{\ast }S_{2}G_{1})  \notag \\
&&-\left[ X_{1}-G_{2}^{\ast }S_{2}C_{1},X_{2}\right] \left[ 
\begin{array}{cc}
A_{1} & 0 \\ 
-C_{2}^{\ast }S_{2}C_{1} & A_{2}
\end{array}
\right] ^{-1}\left[ 
\begin{array}{c}
Z_{1} \\ 
Z_{2}-C_{2}^{\ast }S_{2}G_{1}
\end{array}
\right] .  \label{IFthenAE}
\end{eqnarray}

\subsection{Commutativity of the Limits: Systems in Series}

The matrix inverse appearing in (\ref{IFthenAE}) is easily computed as a
special case of the well-known formula for the inverse of block matrices
(the earliest reference is credited to Banachiewicz \cite{Ban37}, see
subsection \ref{subsec:Schur}, however, like many matrix identities the
origins may be considerably older)

\begin{equation}
\left[ 
\begin{array}{cc}
A_{1} & 0 \\ 
-C_{2}^{\ast }S_{2}C_{1} & A_{2}
\end{array}
\right] ^{-1}=\left[ 
\begin{array}{cc}
A_{1}^{-1} & 0 \\ 
A_{2}^{-1}C_{2}^{\ast }S_{2}C_{1}A_1^{-1} & A_{2}^{-1}
\end{array}
\right] .  \label{inverse}
\end{equation}
This yields the explicit form 
\begin{eqnarray*}
\hat{S}_{\mathrm{ser}} &=&(I+S_{2}C_{1}A_{1}^{-1}C_{1}^{\ast }S_{2}^{\ast
}+C_{2}A_{2}^{-1}C_{2}^{\ast }S_{2}C_{1}A_{1}^{-1}C_{1}^{\ast }S_{2}^{\ast
}+C_{2}A_{2}^{-1}C_{2}^{\ast })S_{2}S_{1} \\
&=&(I+C_{2}A_{2}^{-1}C_{2}^{\ast })S_{2}(I+C_{1}A_{1}^{-1}C_{1}^{\ast })S_{1}
\\
&\equiv &\hat{S}_{2}\hat{S}_{1}.
\end{eqnarray*}
The coupling operator is 
\begin{eqnarray*}
\hat{L}_{\mathrm{ser}} &=&(G_{1}+S_{2}G_{1})-S_{2}C_{1}A_{1}^{-1}Z_{1} \\
&&-C_{2}A_{2}^{-1}C_{2}^{\ast
}S_{2}C_{1}A_{1}^{-1}Z_{1}-C_{2}A_{2}^{-1}Z_{2}+C_{2}A_{2}^{-1}C_{2}^{\ast
}S_{2}G_{1} \\
&=&(G_{2}-C_{2}A_{2}^{-1}Z_{2})+(I+C_{2}A_{2}^{-1}C_{2}^{\ast
})S_{2}(G_{1}-C_{1}A_{1}^{-1}Z_{1}) \\
&\equiv &\hat{L}_{2}+\hat{S}_{2}\hat{L}_{1}.
\end{eqnarray*}
Finally we see that 
\begin{eqnarray*}
\hat{K}_{\mathrm{ser}} &=&R_{1}+R_{2}-G_{2}^{\ast
}S_{2}G_{1}-X_{1}A_{1}^{-1}Z_{1} \\
&&+G_{2}^{\ast
}S_{2}C_{1}A_{1}^{-1}Z_{1}-X_{2}A_{2}^{-1}Z_{2}+X_{2}A_{2}^{-1}C_{2}^{\ast
}S_{2}(G_{1}-C_{1}Z_{1}).
\end{eqnarray*}
We would like to show that this equals $\hat{K}_{1}+\hat{K}_{2}-\hat{L}%
_{2}^{\ast }\hat{S}_{2}\hat{L}_{1}$, now we have 
\begin{multline*}
\hat{K}_{1}+\hat{K}_{2}-\hat{L}_{2}^{\ast }\hat{S}_{2}\hat{L}%
_{1}=R_{1}-X_{1}A_{1}^{-1}Z_{1}+R_{2}-X_{2}A_{2}^{-1}Z_{2} \\
-(G_{2}^{\ast }-Z_{2}^{\ast }A_{2}^{-1\ast
}C_{2}^*)(I+C_{2}A_{2}^{-1}C_{2}^{\ast })S_{2}(G_{1}-C_{1}A_{1}^{-1}Z_{1}),
\end{multline*}
and to compute this we note that $A_{2}=-\frac{1}{2}C_{2}^{\ast
}C_{2}-i\Omega _{2}$ so that 
\begin{eqnarray}
A_{2}^{-1\ast }C_{2}^*(I+C_{2}A_{2}^{-1}C_{2}^{\ast }) &=&A_{2}^{-1\ast
}(I+C_{2}^{\ast }C_{2}A_{2}^{-1})C_{2}^{\ast }  \notag \\
&=&A_{2}^{-1\ast }(A_{2}+C_{2}^{\ast }C_{2})A_{2}^{-1}C_{2}^{\ast }  \notag
\\
&=&A_{2}^{-1\ast }(-A_{2}^{\ast })A_{2}^{-1}C_{2}^{\ast }  \notag \\
&=&-A_{2}^{-1}C_{2}^{\ast },  \label{trick}
\end{eqnarray}
this yields 
\begin{eqnarray*}
\hat{K}_{1}+\hat{K}_{2}-\hat{L}_{2}^{\ast }\hat{S}_{2}\hat{L}_{1}
&=&R_{1}-X_{1}A_{1}^{-1}Z_{1}+R_{2}-X_{2}A_{2}^{-1}Z_{2} \\
&&-G_{2}^{\ast }(I+C_{2}A_{2}^{-1}C_{2}^{\ast
})S_{2}(G_{1}-C_{1}A_{1}^{-1}Z_{1}) \\
&&-Z_{2}^{\ast }A_{2}^{-1}C_{2}^{\ast }S_{2}(G_{1}-C_{1}A_{1}^{-1}Z_{1}).
\end{eqnarray*}
We therefore find that 
\[
\hat{K}_{\mathrm{ser}}-(\hat{K}_{1}+\hat{K}_{2}-\hat{L}_{2}^{\ast }\hat{S}%
_{2}\hat{L}_{1})=\left\{ X_{2}+G_{2}^{\ast }C_{2}+Z_{2}^{\ast }\right\}
A_{2}^{-1}C_{2}^{\ast }S_{2}(G_{1}-C_{1}Z_{1})
\]
however this vanishes identically by the second of identities (\ref
{identities}).

We therefore conclude that the model parameters in (\ref{AEthenIF}) are
identical with those in (\ref{IFthenAE}), therefore the adiabatic
elimination and instantaneous feedforward limit commute.

\subsection{Adiabatic Elimination: In-Loop Device}
\label{sec:sys-in-loop}
Next we want to extend our investigation to situations where we have a
non-trivial feedback arrangement as illustrated in Fig. \ref
{AE:fig:FeedbackArragement}. The question again is whether the two limits
commute.

\begin{figure}[htbp]
\centering
\includegraphics[width=0.50\textwidth]{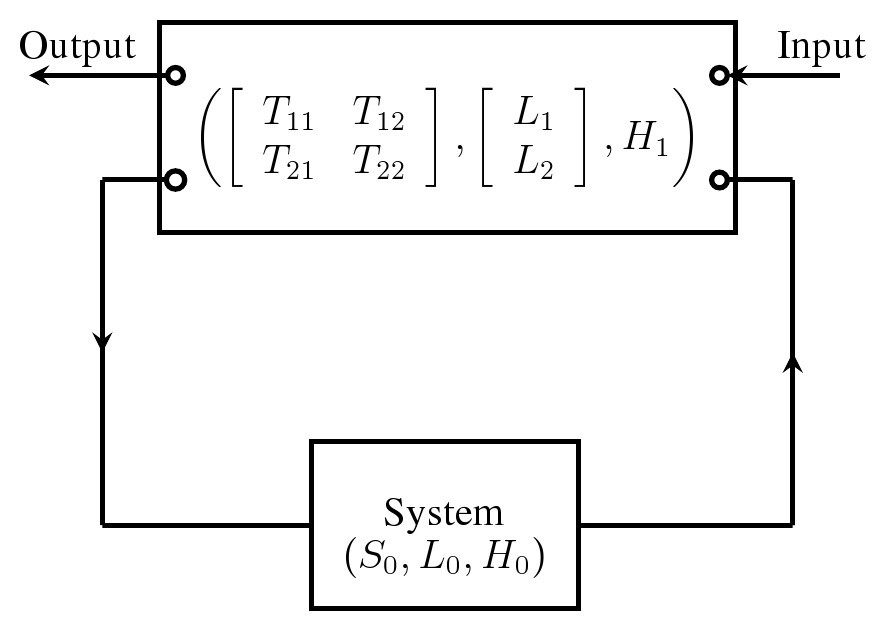} 
%	\label{fig:AE_LFT_System}
\caption{General feedback arrangement: The four port device interacts with
one external input, producing one external output and interacts with an
in-loop device by one internal in- and output field respectively.}
\label{AE:fig:FeedbackArragement}
\end{figure}

We start off with a simple model in-loop, taking the 4-port device to be a
beam splitter, modeled by a unitary matrix $T =[T_{jl}]$ with complex
entries and where coupling operators and the systems Hamiltonian are zero, $%
L_1 = L_2 = H_1 = 0$. We parameterize the in-loop device as 
\begin{eqnarray}
S_0(k) &=& S_0  \notag  \label{AE:eq:SimpleModel1} \\
L_0(k) &=& k \sqrt{\gamma} a_0  \notag \\
K_0(k) &=& - \frac{1}{2} k^2 a_0^\ast \gamma a_0.  \label{AE:eq:SimpleModel2}
\end{eqnarray}
and fix the beam splitter (scattering) matrix as (with $\alpha$ real) 
\begin{equation}
T = \left[ 
\begin{array}{cc}
\alpha & \sqrt{1 - \alpha^2} \\ 
\sqrt{1 - \alpha^2} & - \alpha
\end{array}
\right].
\end{equation}

\begin{figure}[htbp]
\centering
\includegraphics[width=0.40\textwidth]{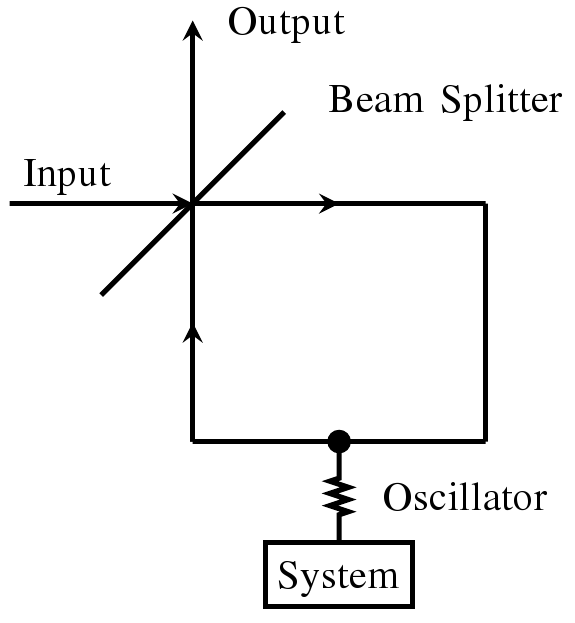}
\caption{Oscillator in-loop}
\label{fig:AE_BS_Oscy_System}
\end{figure}

Thus, the in-loop system consists only of a single oscillator coupled to the
in-loop field and no additional modes coupled to the oscillator. In terms of
operator parameters $S_0,C_0,G_0,A_0,Z_0,X_0,R_0$, see equation (\ref{scaled parameters}), 
\begin{eqnarray*}
S_0(k) &=&S_0 \otimes I \\
L_0(k) &=&k C_0 \otimes a +G_0 \otimes I \\
K_0(k) &=&k^{2} A_0 \otimes a^{\ast }a +k Z_0 \otimes
a^{\ast }+k X_0 \otimes a +R_0 \otimes I
\end{eqnarray*}
acting on the space $\mathfrak h_{\text{sys}}\otimes \mathfrak h_{\text{osc}%
}$, we see that 
\begin{eqnarray*}
S_{0} &=&S_{0} \\
A_{0} &=&-\frac{1}{2}\gamma  \\
C_{0} &=&\sqrt{\gamma } \\
Z_{0} &=&X_{0}=R_{0}=G_{0}=0.
\end{eqnarray*}
The coefficients for the single input single output device after taking the
instantaneous feedback limit of the arrangement of Fig.  \ref
{AE:fig:FeedbackArragement} were derived by Gough and James\cite{GoughJamesCMP09} and are
given by 
\begin{eqnarray}
S_{\text{red}} &=&T_{11}+T_{12}S_{0}(I-T_{22}S_{0})^{-1}T_{21}  \notag \\
L_{\text{red}} &=&T_{12}(I-T_{22}S_{0})^{-1}L_{0}  \notag \\
H_{\text{red}} &=&K_{0}-L_{0}^{\ast }S_{0}(I-T_{22}S_{0})^{-1}L_{0}.
\label{GJ_loop}
\end{eqnarray}
For the model (\ref{AE:eq:SimpleModel2}), the limit coefficients after taking the adiabatic elimination limit (see Theorem \ref{thm: adiabatic elimination}) are
given by 
 \begin{eqnarray}
\hat{S}_{0} &=&\left( I+C_{0}A_{0}^{-1}C_{0}^{\ast }\right) S_{0}=-S_{0}, 
\notag \\
\hat{L}_{0} &=&G_{0}-C_{0}A_{0}^{-1}Z_{0}=0,  \notag \\
\hat{K}_{0} &=&R_{0}-X_{0}A_{0}^{-1}Z_{0}=0.  \label{AE:eq:AEModellFirst}
\end{eqnarray}
Substituting into (\ref{GJ_loop}) we find that the reduced coefficients after the instantaneous feedback limit for the
model (\ref{AE:eq:AEModellFirst}) are 
\begin{eqnarray}
\hat{S} &=&\alpha +\sqrt{1-\alpha ^{2} }(- S_{0}) \frac{1} {1-(-\alpha)(- S_{0}) } \sqrt{1 -\alpha^2} =\frac{\alpha -S_{0}%
}{1-\alpha S_{0}}  \label{AE:eq:S_AEFirstIFSecond} \\
\hat{L} &=&0  \notag \\
\hat{K} &=&0.  \notag
\end{eqnarray}

We now exchange the order in which we perform the limits.
The instantaneous feedback limit of the model before taking the adiabatic
elimination limit yields: 
\begin{eqnarray*}
\tilde S (k) &=& T_{11} + T_{12} S_0 \left( I - T_{22}S_0\right)^{-1}T_{21}
= \alpha + \left( 1 - \alpha^2\right)S_0 \frac{1}{1+\alpha S_0} \\
\tilde L (k) &=& k \sqrt{1 - \alpha^2} \frac{1}{1 + \alpha S_0}\sqrt{\gamma}
a_0 \\
\tilde K (k) &=& K_0(k) - L_0(k)^\ast \frac{S_0T_{22}}{1-S_0T_{22}}L_0(k) =
k^2 a_0^\ast \left( -\frac{1}{2}\gamma + \gamma \frac{\alpha S_0}{1 + \alpha
S_0}\right)a_0.
\end{eqnarray*}
where the operator parameters are 
\begin{eqnarray*}
A &=& -\frac{\gamma}{2}\frac{1 - \alpha S_0}{1 + \alpha S_0} \\
C &=& \frac{\sqrt{1 - \alpha^2}\sqrt{\gamma}}{1 + \alpha S_0} \\
S &=& \alpha + \frac{\left( 1 - \alpha^2\right)S_0}{1+\alpha S_0} \\
G &=& X = Z = R = 0
\end{eqnarray*}
The A.E. of the I.F. limit model then results in (here $|1+\alpha S_0|^2=(1+\alpha S_0^*)(1+\alpha S_0)$)
\begin{eqnarray}
\hat{S} &=&\left( 1-\frac{(1-\alpha ^{2})\gamma }{|1+\alpha S_{0}|^{2}}\frac{%
2}{\gamma }\frac{1+\alpha S_{0}}{1-\alpha S_{0}}\right) \left( \alpha +\frac{%
(1-\alpha ^{2})S_{0}}{1+\alpha S_{0}}\right) 
 \notag \\
&=&\frac{\left( \alpha S_{0}^{\ast }-1\right) \left( 1+\alpha S_{0}\right) }{%
\left( 1+\alpha S_{0}^{\ast }\right) \left( 1-\alpha S_{0}\right) }\left( 
\frac{\alpha +S_{0}}{1+\alpha S_{0}}\right)  \notag \\
&=&\frac{\alpha -S_{0}}{1-\alpha S_{0}}.
\label{AE:eq:S_IFFirstAESecond}
\end{eqnarray}

We see that the limits  do in fact commute since we obtain the same operator $\hat S$ in (\ref{AE:eq:S_AEFirstIFSecond}) and (\ref{AE:eq:S_IFFirstAESecond}), likewise for the operators $\hat L$ and $\hat K$. The apparently miraculous agreement comes as a general feature that will be observed in more complex networks.  Our approach will be to encode both these limits
as instances of a Schur complement operation: the miraculous agreements that
one encounters in a case-by-case study are in fact just
by-product of these operations.

If $S_0=1$ then in quantum optics the system  $(S_0(k),L_0(k),K_0(k))$ represents a  one-sided optical cavity  in which  the coupling coefficient $k\sqrt{\gamma}$ of the partially transmitting cavity mirror is large (for large $k$). The calculations of this section show that for large $k$ the network in Fig.~\ref{fig:AE_BS_Oscy_System} can be consistently approximated by an {\em effective} device that  shifts the phase of the output field with respect to the input field by an amount determined by the parameters of the cavity and beam splitter. Alternatively, one can also think of the network of  Fig.~\ref{fig:AE_BS_Oscy_System} as approximately implementing a phase shifting device.

\section{Adiabatic Elimination Within General Networks}
\label{sec:AE-QFN}
The situation of two systems in cascade, as depicted in Fig. \ref
{fig:series2}, is the simplest form of a nontrivial quantum feedback
network. We remark that at no stage of the calculations did we assume that
the operators describing the first system commuted with those of the second
system. Indeed, the series product is valid even if we do not assume that we
are dealing with separate cascaded systems and is applicable to the problem
of feedback into the same system.

In Fig. \ref{fig:qfn} we describe a somewhat more engorged quantum
feedback network featuring feedback and feedforward interconnections. For each component of
the network, we will have the same multiplicity for the input fields as the
output fields, though we split up the inputs and outputs geometrically to
indicate different physical connections for the fields. The unitary $S$ for
a given component now additionally implies that we can use the component to
mix the input fields, with a beam-splitter being the very special case where
the entries of $S$ are just complex constants. This feature introduces the
possibility of topologically non-trivial feedback loops that were not
present in the simple situations of direct feedforward or feedback occurring
for systems in series.

We now aim to investigate the procedure of adiabatic elimination of fast
degrees of freedom from components in a general quantum feedback network
and, in particular, answer the question of whether this commutes with the
Markov limit in which we take vanishing time lags for the various internal
fields in the network. The adiabatic elimination of oscillators for
components in series will then be a very specific case of this general
theory.

\begin{figure}[h]
\centering
\includegraphics[width=0.60\textwidth]{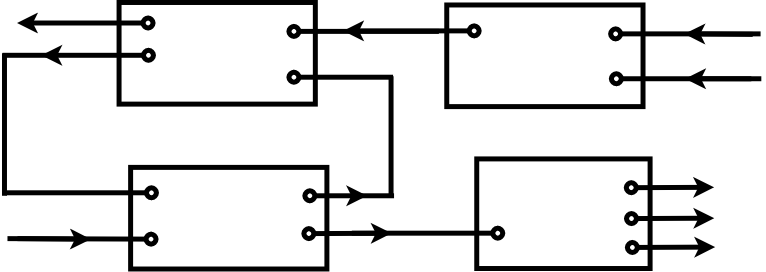}
\caption{Quantum feedback network}
\label{fig:qfn}
\end{figure}

The essentially mathematical element in the investigation will be that both
the adiabatic elimination limit and the instantaneous feedback limit for a
general quantum feedback network are actually instances of a Schur
complement of the It\={o} matrix $\mathbf{G}$.

\subsection{The Schur Complement}

We begin by recalling some of the basic definitions and notations relevant for Schur complements. For general reviews, see the survey article by Oullette\cite{DVO81} or 
the book chapter by Horn and Zhang\cite{Zhang}. We shall elaborate on several of the well-known results presented in the reviews, largely to take account of the fact that
we are dealing with block-partitioned matrices with operator entries. In particular we give some minor technical extensions where we are explicit about the domains, 
kernel spaces and image spaces on which the operators act.

The Schur complement of an $\left( n+m\right) $ square matrix $M=\left[ 
\begin{array}{cc}
A & B \\ 
C & D
\end{array}
\right] $ relative to the $m\times m$ sub-block $A$ is defined to be 
\[
M/A=D-CA^{-1}B
\]
under the assumption that $A$ is invertible. We note the following
elementary formula, due to Banachiewicz \cite{Ban37}, for invertible $M$%
\[
\left[ 
\begin{array}{cc}
A & B \\ 
C & D
\end{array}
\right] ^{-1}=\left[ 
\begin{array}{cc}
A^{-1}+A^{-1}B\left( M/A\right) ^{-1}CA^{-1} & -A^{-1}B\left( M/A\right)
^{-1} \\ 
-\left( M/A\right) ^{-1}CA^{-1} & \left( M/A\right) ^{-1}
\end{array}
\right] .
\]

\begin{definition}
A matrix $M^{-}$ is a generalized inverse for a square matrix $M$ if we have 
$MM^{-}M=M$. The generalized Schur complement of $M = \left[ 
\begin{array}{cc}
A & B \\ 
C & D
\end{array}
\right]$ is then defined to be 
\begin{equation}
M/A = D - CA^- B .
\end{equation}
\end{definition}

\begin{lemma}
\label{lm:GS-exist} The generalized Schur complement $M/A$ is well-defined and independent of
the choice of the generalized inverse $A^-$ so long as we have the following
inclusions of image spaces $\mathrm{im}\,{B} \subseteq \mathrm{im}\,{A}$ and
kernel spaces $\mathrm{ker}\,{A} \subseteq \mathrm{ker}\,{C}$.
\end{lemma}

Note that $\mathrm{im}\,{B} \subseteq \mathrm{im}\,{A}$ occurs if and only if $%
\mathrm{ker}\,{A^\ast} \subseteq \mathrm{ker}\,{B^\ast}$. (Recall that the
image, or column space, of a matrix is the span of its columns, or more
generally $\mathrm{im}\,(M)= \mathrm{ker}\,(M^{\ast })^{\perp }$.)

\begin{lemma}
\label{lm:RL-inverse} For two matrices $M$ and $N$ and some generalized inverse $M^-$ of $M$ we
have that 
\[
M M^- N = N \text{ if } \mathrm{im}\,{N} \subset \mathrm{im}\,{M}
\]
and 
\[
NM^-M = N \text{ if } \mathrm{ker}\,{M} \subset \mathrm{ker}\,{N}.
\]
\end{lemma}
For the proofs of these lemmata, see Horn and Zhang\cite{Zhang}; they are a straightforward consequence of the definition of a generalized inverse and the postulated image/kernel inclusions.

\label{subsec:Schur}
The Schur complement and Lemmata \ref{lm:GS-exist} and \ref{lm:RL-inverse} above may be generalized to matrices with operator entries. %
 Let $M$ be a bounded invertible operator on a Hilbert
space $\mathfrak{H}$ and let us fix a decomposition $\mathfrak{H}=\oplus
_{j\in \mathfrak{I}}\mathfrak{H}_{j}$ for some finite index set $\mathfrak{I}
$. We denote by $x_{j}$ the component of a vector $x\in \mathfrak{H}$ in $%
\mathfrak{H}_{j}$, and $M_{jl}$ the block component of $M$ mapping from $%
\mathfrak{H}_{j}$ to $\mathfrak{H}_{l}$. For $A=\left\{ a_{1},\cdots
,a_{n}\right\} $ and $B=\left\{ b_{1},\cdots ,b_{m}\right\} $ non-empty
subsets of $\mathfrak{I}$ we write 
\[
x_{A}=\left[ 
\begin{array}{c}
x_{a_{1}} \\ 
\vdots \\ 
x_{a_{m}}
\end{array}
\right] ,M_{A,B}=\left[ 
\begin{array}{ccc}
M_{a_{1}b_{1}} & \cdots & M_{a_{1}b_{m}} \\ 
\vdots & \ddots & \vdots \\ 
M_{a_{n}b_{m}} & \cdots & M_{a_{n}b_{m}}
\end{array}
\right] .
\]
The single equation $Mx=u$ then corresponds to the coarsest block form $%
M_{\mathfrak{I},\mathfrak{I}}x_{\mathfrak{I}}=u_{\mathfrak{I}}$. In contrast, the full system of equations $\sum_{l \in
\mathfrak{I}}M_{jl}x_{l}=u_{j}$ gives the finest block form. More generally, we may
examine intermediate partitions. Let $A$ and $B$ be non-trivial (i.e.,
non-empty, proper) subsets of $\mathfrak{I}$ then the equation $Mx=u$ may be written as 
\begin{equation}
\left[ 
\begin{array}{cc}
M_{A,B} & M_{A,B^{\prime }} \\ 
M_{A^{\prime },B} & M_{A^{\prime },B^{\prime }}
\end{array}
\right] \left[ 
\begin{array}{c}
x_{B} \\ 
x_{B^{\prime }}
\end{array}
\right] =\left[ 
\begin{array}{c}
u_{A} \\ 
u_{A^{\prime }}
\end{array}
\right] , ,  \label{Mx=u}
\end{equation}
where $A^{\prime}$ denotes the complement of set $A$ in $\mathfrak{I}$, and
the inverse relation is 
\begin{equation}
\left[ 
\begin{array}{c}
x_{B} \\ 
x_{B^{\prime }}
\end{array}
\right] =\left[ 
\begin{array}{cc}
(M^{-1})_{B,A} & (M^{-1})_{B,A^{\prime }} \\ 
(M^{-1})_{B^{\prime },A} & (M^{-1})_{B^{\prime },A^{\prime }}
\end{array}
\right] \left[ 
\begin{array}{c}
u_{A} \\ 
u_{A^{\prime }}
\end{array}
\right] .  \label{x=M-1u}
\end{equation}
We now recall the definition of the generalized Schur complement, sometimes also known
as the shorted operator, in the case where $M$ need not be invertible.

\begin{definition}
\label{def:Schur} Let $A$ and $B$ be non-trivial subsets of the index $%
\mathfrak{I}$, and let $C$ be a non-trivial subset of $A$, and $D$ be a
nontrivial subset of $B$. Furthermore take $|A|=|B|$ and $|C|=|D|$. Suppose
that the sub-block $M_{C,D}$ possesses a generalized inverse denoted by $%
(M_{C,D})^-$, then the Schur complement of $M_{A,B}$ relative to $M_{C,D}$
is defined to be 
\[
M_{A,B}/M_{C,D}=M_{A/C,B/D}-M_{A/C,D}(M_{C,D})^-M_{C,B/D}.
\]
In the special case where $A=B=%
\mathfrak{J}$, we shall simply write $M/M_{C,D}$ for $M_{\mathfrak{J},%
\mathfrak{J}}/M_{C,D}$.
\end{definition}

The generalized Schur complement is well-defined and independent of the
choice of generalized inverse so long as the column space $\mathrm{im}%
(M_{C,B/D})$ is contained in $\mathrm{im}(M_{C,D})$, and $\mathrm{ker}
(M_{C,D})$ is contained in $\mathrm{ker} (M_{A/C,D})$. In particular, if the conditions of the Lemma \ref{lm:GS-exist} are met then
we may fix a particular generalized inverse such as the Moore-Penrose
inverse. We also remark that we may readily extend the above notation to the
case where different direct-sum decompositions of $\mathfrak{H}$\ are used
for the columns and rows of $M$.

We shall also require the extension of the Banachiewicz formula to generalized inverses. The proof for finite rank matrix operators is due to Marsaglia and Styan \cite{MS}, 
and may be found as Theorem 4.6 in the article of Ouellette \cite{DVO81}. In the next lemma, we strengthen this to deal with general Hilbert space operators.

\begin{lemma}
\label{lem:gen inverse formula} Let $M$ be partitioned according to $M = %
\left[ 
\begin{array}{cc}
A & B \\ 
C & D
\end{array}
\right]$. We suppose that $\mathrm{im}\,{B} \subseteq \mathrm{im}\,{A}$, $%
\mathrm{ker}\,{A} \subseteq \mathrm{ker}\,{C}$, and therefore the generalized
Schur complement $X = M/A= D - C A^- B $ is well-defined and independent of
the choice of generalized inverse $A^-$ to $A$. Then the generalized inverse
of $M$ is given by 
\[
M^- = \left[ 
\begin{array}{cc}
A^- + A^- B X^- CA^- & -A^- B X^- \\ 
- X^- C A^- & X^-
\end{array}
\right].
\]
\end{lemma}

\begin{proof}
% of lemma {lem:gen inverse formula}
We multiply out $M M^- M$ in block form. The top left block will be 
\begin{eqnarray*}
(M M^- M)_{11} &=& A A^- A + A A^- B X^- C A^- A - B X^- C A^- A \\
&&- A A^- B X^- C + BX^- C \\
&=& A A^- A + A A^- B X^- C(A^- A - 1) - BX^- C ( A^- A -1) \\
&=& A A^- A + (A A^- -1 )BX^- C(A^- A - 1) \\
&=& A,
\end{eqnarray*}
where the last step follows because $A A^- A = A$ and $(A A^- -1 )B = 0$
under the assumptions that $\mathrm{im}\,{B} \subseteq \mathrm{im}\,{A}$.
Similarly 
\begin{eqnarray*}
(M M^- M)_{12} &=& B + (A A^- - 1)B + \\
&& (A A^- - 1)B X^- C A^- B (A A^- - 1)B X^- D \\
&=& B,
\end{eqnarray*}
since under the assumption $\mathrm{im}{B} \subseteq \mathrm{im}{A}$ we have
that $A A^- B = B$ and so $\left( A A^- - 1 \right)B = 0$; 
\begin{eqnarray*}
(M M^- M)_{21} &=& C + C (A^- A - 1) + (C A^- B -D)X^- C (A^- A - 1) \\
&=& C,
\end{eqnarray*}
because of the assumption $\mathrm{ker}\,{A} \subseteq \mathrm{ker}\,{C}$ we
have $C A^- A = C$ and $C(A^- A - 1)=0$; and 
\begin{eqnarray*}
(M M^- M)_{22} &=& D - (D - C A^- B) - (D - C A^- B )X^- ( D - C A^- B) \\
&=& D - X + XX^- X \\
&=& D,
\end{eqnarray*}
since $X = M/A = D - C A^- B$ and $X X^- X = X$. Collecting these results we
have that $M M^- M = M$, as required.
\end{proof}

Now, as  a corollary to Lemma \ref{lem:gen inverse formula} we obtain
the generalized Banachiewicz formula:
\begin{eqnarray*}
M^{-} _{B,A} &=& M_{A,B}^{-}+ M_{A,B} ^{-}M_{A,B^{\prime }}(
M/M_{A,B}) ^{-}M_{A^{\prime },B} M_{A,B}^{-}, \\
M^{-}_{B,A^{\prime }} &=&- M_{A,B}^{-}M_{A,B^{\prime }}( M/M_{A,B})
^{-}, \\
M^{-}_{B^{\prime },A} &=&- (M/M_{A,B})^{-}M_{A^{\prime },B} M_{A,B}
^{-}, \\
M^{-}_{B^{\prime },A^{\prime }} &=& (M/M_{A,B})^{-}.
\end{eqnarray*}

We now wish to establish the property of commutativity of successive Schur complementation as this shall be the main technical result required in this paper.

\begin{lemma}[Successive complementation rule]
\label{lem:scrule} Suppose that $A,B,C$ is a partition of the index set $%
\mathfrak{I}$ (that is, $A,B,C$ are disjoint non-empty subsets whose union
is $\mathfrak{J}$) then, whenever the generalized Schur complements are well-defined, we have
the rule 
\begin{eqnarray}
M/M_{B\cup C,B\cup C} &=&(M/M_{C,C})/(M/M_{C,C})_{B,B}  \notag \\
&=&(M/M_{B,B})/((M/M_{B,B}))_{C,C}.  \label{Scrule}
\end{eqnarray}
\end{lemma}

For the case of matrices over a field where the inverses exist, the first equality in (\ref{Scrule}) is an instance of the Crabtree-Haynsworth quotient formula \cite{DVO81}.
The extension of the quotient formula to generalized inverses for matrices over a field was given by Carson, Haynsworth and Markham \cite{CHM} under some rank conditions, see Theorem 4.8 in the review by Ouellete \cite{DVO81}. However, since here we are dealing with matrices with Hilbert space operator entries rather than just matrices over a field, we need to extend this result accordingly.  To this end, below we independently prove a generalization of the algebraic content of the theorem to matrices with Hilbert space operator entries, modulo the conditions for these Schur complements to be well-defined which we defer to Lemma \ref{new} in the Appendix.

\begin{proof}
% proof of {lem:scrule}
Assume that the conditions of Lemma \ref{new} are in place. Let us first compute $(M/M_{C,C})/(M/M_{C,C})_{B,B}$: 
\begin{eqnarray*}
M/M_{C,C} &=&\left[ 
\begin{array}{ccc}
M_{A,A} & M_{A,B} & M_{A,C} \\ 
M_{B,A} & M_{B,B} & M_{B,C} \\ 
M_{C,A} & M_{C,B} & M_{C,C}
\end{array}
\right] /M_{C,C} \\
&=&\left[ 
\begin{array}{cc}
M_{A,A} & M_{A,B} \\ 
M_{B,A} & M_{B,B}
\end{array}
\right] -\left[ 
\begin{array}{c}
M_{A,C} \\ 
M_{B,C}
\end{array}
\right] (M_{C,C})^- 
\begin{array}{cc}
\lbrack M_{C,A} & M_{C,B}]
\end{array}
\\
&=&\left[ M_{L,R}-M_{L,C}(M_{C,C})^-M_{C,R}\right] _{R,L\in \left\{
A,B\right\} }
\end{eqnarray*}
so a second Schur complementation leads to 
\begin{multline*}
(M/M_{C,C})/(M/M_{C,C})_{B,B}=M_{A,A}-M_{A,C}(M_{C,C})^-M_{C,A} \\
-(M_{A,B}-M_{A,C}(M_{C,C})^-M_{C,B})\Xi (M_{B,A}-M_{B,C}(M_{C,C})^-M_{C,A}),
\end{multline*}
where we write $\Xi =(M_{B,B}-M_{B,C}(M_{C,C})^-M_{C,B})^-$ for shorthand.
We then compute $M/M_{B\cup C,B\cup C}$  
\begin{multline*}
M/M_{B\cup C,B\cup C}=\left[ 
\begin{array}{ccc}
M_{A,A} & M_{A,B} & M_{A,C} \\ 
M_{B,A} & M_{B,B} & M_{B,C} \\ 
M_{C,A} & M_{C,B} & M_{C,C}
\end{array}
\right] /\left[ 
\begin{array}{cc}
M_{B,B} & M_{B,C} \\ 
M_{C,B} & M_{C,C}
\end{array}
\right] \\
=M_{AA}-\left[ 
\begin{array}{cc}
M_{A,B} & M_{A,C}
\end{array}
\right] \left[ 
\begin{array}{cc}
M_{B,B} & M_{B,C} \\ 
M_{C,B} & M_{C,C}
\end{array}
\right] ^-\left[ 
\begin{array}{c}
M_{B,A} \\ 
M_{C,A}
\end{array}
\right] ,
\end{multline*}
however, the inverse can be computed explicitly using the Banachiewicz
formula as 
\[
\left[ 
\begin{array}{cc}
\Xi & -\Xi M_{B,C}(M_{C,C})^- \\ 
-(M_{C,C})^-M_{C,B}\Xi & (M_{C,C})^-+(M_{C,C})^{-}M_{C,B}\Xi
M_{B,C}(M_{C,C})^-
\end{array}
\right] .
\]
Multiplying out the block matrix readily leads to the same expression
already obtained for $(M/M_{C,C})/(M/M_{C,C})_{B,B}$. The second equality in
(\ref{Scrule}) follows from Lemma \ref{new} and by interchanging $B$ and $C$.
\end{proof}

\subsection{Instantaneous Feedback Limit as Schur Complement}

Suppose that we are given a collection of components which have separate descriptions 
$( S_{j},L_{j},K_{j}) $ for $j=1,2,\cdots, c$. We may collect them into a single model $( S,L,K) $ given by 
\[
S=\left[ 
\begin{array}{cccc}
S_{1} & 0 & \cdots & 0 \\ 
0 & S_{2} & \cdots & 0 \\ 
\vdots & \vdots & \ddots & 0 \\ 
0 & 0 & 0 & S_{c}
\end{array}
\right] ,\quad L=\left[ 
\begin{array}{c}
L_{1} \\ 
L_{2} \\ 
\vdots \\ 
L_{c}
\end{array}
\right] ,\quad K=\sum_{j=1}^{c}K_{j}.
\]
This just describes the open-loop system where no connections have been made
and all input and output fields are therefore external.

To obtain the closed loop description we have to give a list of which
outputs are to be fed back in as inputs. Algebraically, this comes down to
assembling a total multiplicity space $\mathfrak{K}=\oplus _{j=1}^{c}%
\mathfrak{K}_{j} $ and a joint system space $\mathfrak{h}=\otimes _{j=1}^{c}%
\mathfrak{h}_{j}$. In this way we obtain a network matrix $\mathbf{G}$ on $%
\mathfrak{h}\otimes (\mathbb{C}\oplus \mathfrak{K})$ from the component
matrices $\mathbf{G}_{j}$ on $\mathfrak{h}_{j}\otimes (\mathbb{C}\oplus \mathfrak{K}%
_{j})$.

Once the connections have been specified, we arrive at a decomposition 
\[
\mathfrak{K}=\mathfrak{K}_{\rm e}\oplus \mathfrak{K}_{\rm i}
\]
where $\mathfrak{K}_{\rm e}$ lists all the external fields and $\mathfrak{K}_{\rm i}$
lists all the internal fields. This decomposition induces a decomposition of 
$\mathfrak{H}$ as 
\[
\mathfrak{h}\otimes (\mathbb{C}\oplus \mathfrak{K})=\left[ \mathfrak{h}%
\otimes (\mathbb{C}\oplus \mathfrak{K}_{\rm e})\right] \oplus \left[ \mathfrak{h}%
\otimes \mathfrak{K}_{\rm i}\right]
\]
and with respect to this decomposition, the It\={o} matrix may be
partitioned based on internal (`${\rm i}$') and external (`${\rm e}$') components as (see Gough and James \cite{GoughJamesCMP09} for details)
\begin{equation}
\mathbf{G}=\left[ 
\begin{array}{cc}
G_{\rm ee} & G_{\rm ei} \\ 
G_{\rm ie} & G_{\rm ii}
\end{array}
\right],  \label{G partitioned e and i}
\end{equation}
In Fig.~\ref{fig:externalinternal} we sketch the picture that emerges when
we subsume all the external fields together and all the internal fields
together as single channels.

\begin{figure}[htbp]
\centering
\includegraphics[width=0.40\textwidth]{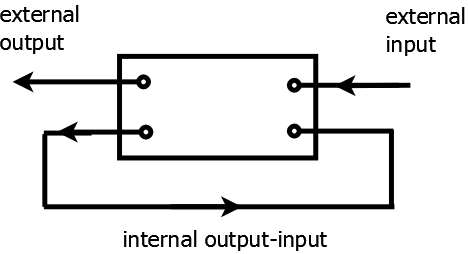}
\caption{Feedback situation}
\label{fig:externalinternal}
\end{figure}

In the instantaneous feedback limit we find that the reduced model is
described by the It\={o} matrix $\mathcal{F} \mathbf{G}\in \mathcal{B}(\mathfrak{h}%
\otimes (\mathbb{C}\oplus \mathfrak{K}_{\rm e}))$ given by the Schur complement 
\begin{equation}
\mathcal{F} \mathbf{G}=G_{\rm ee}-G_{\rm ei}(G_{\rm ii})^{-1}G_{\rm ie},  \label{f-reduction}
\end{equation}
provided that $G_{\rm ii}$ exists. We remark that the original version of this formula \cite{GoughJamesCMP09} involved the related
\emph{model matrix} $\mathbf{U}=\left[ 
\begin{array}{cc}
K & -L^{\ast }S \\ 
L & S
\end{array}
\right] $ rather than $\mathbf{G}$ and the corresponding feedback reduction map was
the fractional linear transformation $\mathcal{F}%
\mathbf{U}=U_{\rm ee}+U_{\rm ei}(1-U_{\rm ii})^{-1}U_{\rm ie}$. In both cases, the condition that $%
G_{\rm ii}=S_{\rm ii}-I$ be invertible is necessary for the feedback network to be
well-posed.

\begin{remark}
\label{rem:strict_Hurwitz}
We note that models studied here all satisfy a Hurwitz stability condition, though not necessarily in the strict sense. In general, the feedback reduction need not preserve the strictly Hurwitz property, and we may obtain conditionally stable modes through interconnection.
\end{remark}

\subsection{Adiabatic Elimination as Schur Complement}
We now re-examine the adiabatic elimination of oscillators. For finite $k$
we consider the It\={o} matrix 
\[
\mathbf{G} \left( k\right) =\left[ 
\begin{array}{cc}
K\left( k\right) & -L\left( k\right) ^{\ast }S \\ 
L\left( k\right) & S-I
\end{array}
\right]
\]
where we write the scaled operators (\ref{scaled parameters}) as 
\begin{eqnarray*}
K\left( k\right) &=&\left[ I,ka^{\ast }\right] \left[ 
\begin{array}{cc}
R & X \\ 
Z & A
\end{array}
\right] \left[ 
\begin{array}{c}
I \\ 
ka
\end{array}
\right] , \\
L\left( k\right) &=&[I,ka^{\ast }]\left[ 
\begin{array}{cc}
G & C \\ 
0 & 0
\end{array}
\right] \left[ 
\begin{array}{c}
I \\ 
ka
\end{array}
\right] , \\
S &=&[I,ka^{\ast }]\left[ 
\begin{array}{cc}
S & 0 \\ 
0 & 0
\end{array}
\right] \left[ 
\begin{array}{c}
I \\ 
ka
\end{array}
\right] .
\end{eqnarray*}
Recalling Remark \ref{rm:osc-notation} , it is now convenient to use the decomposition $\mathfrak{h}=%
\mathfrak{\hat{h}}\oplus \mathfrak{h}_{f}$ (here $\mathfrak{h}_{f}$ denotes the subspace of the fast oscillator modes) to write 
\[
\mathfrak{h}\otimes (\mathbb{C}\oplus \mathfrak{K})=\left[ \mathfrak{\hat{h}}%
\otimes (\mathbb{C}\oplus \mathfrak{K})\right] \oplus \left[ \mathfrak{h}_{f}\otimes (\mathbb{C}\oplus \mathfrak{K})\right]
\]
and with respect to this decomposition we may now write 
\begin{eqnarray}
\mathbf{G} \left( k\right) &=& k^2 a^* g_{ff} a + k a^* g_{fs} + k g_{sf}a + g_{ss} \notag \\
&\stackrel{\Delta}{=} & [I,ka^{\ast }]\left[ 
\begin{array}{cc}
g_{ss} & g_{sf} \\ 
g_{fs} & g_{ff}
\end{array}
\right] \left[ 
\begin{array}{c}
I \\ 
ka
\end{array}
\right]  \label{G(k) decomp in a}
\end{eqnarray}
and we set
\begin{equation}
g=\left[ 
\begin{array}{cc}
g_{ss} & g_{sf} \\ 
g_{fs} & g_{ff}
\end{array}
\right]. \label{g}
\end{equation}
It is easy to see
that $g$\ is given by 
\begin{eqnarray*}
g_{ss} &=&\left[ 
\begin{array}{cc}
R & -G^{\ast }S \\ 
G & S-I
\end{array}
\right] ,g_{sf}=\left[ 
\begin{array}{cc}
X & 0 \\ 
C & 0
\end{array}
\right] , \\
g_{fs} &=&\left[ 
\begin{array}{cc}
Z & -C^{\ast }S \\ 
0 & 0
\end{array}
\right] ,g_{ff}=\left[ 
\begin{array}{cc}
A & 0 \\ 
0 & 0
\end{array}
\right] .
\end{eqnarray*}
The It\={o} matrix corresponding to the limit operators $(\hat{S},\hat{%
L},\hat{K})$ in (\ref{Limit parameters}) is then 
\[
\hat{\mathbf{G}}=\left[ 
\begin{array}{cc}
\hat{K} & -\hat{L}^{\ast }\hat{S} \\ 
\hat{L} & \hat{S}
\end{array}
\right] =\left[ 
\begin{array}{cc}
R-XA^{-1}Z & -G^{\ast }S+XA^{-1}C^{\ast }S \\ 
G-CA^{-1}Z & S+CA^{-1}C^{\ast }S
\end{array}
\right]
\]
where we use the identity $-\hat{L}^{\ast }\hat{S}=-G^{\ast
}S+XA^{-1}C^{\ast }S$ in the upper right corner which relies on the trick (%
\ref{trick}) along with the identities (\ref{identities}). We then observe
that 
\[
\hat{\mathbf{G}}\equiv g_{ss}-g_{sf}\left( g_{ff}\right) ^-g_{fs}=g/g_{ff}
\]
which is the generalized Schur complement based on the Moore-Penrose inverse 
\[
\left[ 
\begin{array}{cc}
A & 0 \\ 
0 & 0
\end{array}
\right] ^-=\left[ 
\begin{array}{cc}
A^{-1} & 0 \\ 
0 & 0
\end{array}
\right] .
\]
Indeed, given the specific form here we see from the remarks after the
Definition \ref{def:Schur} that any generalized inverse may be used here.
We may then define the adiabatic elimination operator as $\mathcal{A}:\mathbf{G}\left(
k\right) \mapsto \hat{\mathbf{G}}=g/g_{ff}$.

\section{Commutativity of the Limits in General Networks}
\label{sec:AE-IF-commutativity}

Our first step is to see how the instantaneous feedback limit sits with the
adiabatic limit starting from a general model with fast oscillators and
internal connections which we wish to eliminate.

We have seen from (\ref{G(k) decomp in a}) that the It\={o} matrix $\mathbf{G}\left(
k\right) $ may be written as $\mathbf{G}\left( k\right) =[I,ka^{\ast }]g\left[ 
\begin{array}{c}
I \\ 
ka
\end{array}
\right] $ with $g$ given by (\ref{g}). Suppose that the input fields can be partitioned into
internal and external fields that corresponds to a partitioning of $S$ as
\[
S=\left[\begin{array}{cc} S_{\rm ee} & S_{\rm ei} \\ S_{\rm ie} & S_{\rm ii} \end{array} \right],
\]
where $S_{\rm ii}$ is a square matrix pertaining to the scattering of  the internal fields to themselves,  $S_{\rm ee}$ is a square matrix pertaining to the scattering of  the external fields to themselves, while $S_{\rm ei}$ and $S_{\rm ie}$ pertains to a scattering of internal fields to external fields, and vice-versa, respectively.  We also partition $C$ and $G$ accordingly as
$$
C=\left[\begin{array}{cc} C_{\rm e} \\ C_{\rm i} \end{array}\right]; G=\left[\begin{array}{cc} G_{\rm e} \\ G_{\rm i} \end{array}\right].
$$
If we wish to decompose this with
respect to the external and internal field labels, then we obtain 
\[
\mathbf{G} \left( k\right) \equiv \left[ 
\begin{array}{cc}
G_{\rm ee}\left( k\right) & G_{\rm ei}\left( k\right) \\ 
G_{\rm ie}\left( k\right) & G_{\rm ii}\left( k\right)
\end{array}
\right]
\]
similar to (\ref{G partitioned e and i}). We note that these blocks will
necessarily have the following structure 
\begin{eqnarray*}
G_{\rm ee}\left( k\right) &\equiv &[I,ka^{\ast }]g_{\rm ee}\left[ 
\begin{array}{c}
I \\ 
ka
\end{array}
\right] , \\
G_{\rm ei}\left( k\right) &\equiv &[I,ka^{\ast }]g_{\rm ei}, \\
G_{\rm ie}\left( k\right) &\equiv &g_{\rm ee}\left[ 
\begin{array}{c}
I \\ 
ka
\end{array}
\right] , \\
G_{\rm ii}\left( k\right) &=&g_{\rm ii},
\end{eqnarray*}
with 
\begin{eqnarray*}
g_{\rm ee} &=\left[\begin{array}{cc} R_1 & M_1 \\ G_1 & S_{\rm ii}-I \end{array} \right];\,
g_{\rm ei} =\left[\begin{array}{cc} X_1 & 0 \\ C_{\rm i} & 0 \end{array} \right];\\
g_{\rm ie} &=\left[\begin{array}{cc} Z_1 & -C^*S_{\rm i} \\ 0 & 0 \end{array} \right];\, 
g_{\rm ii} =\left[\begin{array}{cc} A& 0 \\ 0 & 0 \end{array} \right],
\end{eqnarray*}
and 
\begin{align*}
S_{\rm e} &= \left[\begin{array}{c} S_{\rm ee} \\ S_{\rm ie} \end{array} \right],\;
S_{\rm i} = \left[\begin{array}{c} S_{\rm ei} \\ S_{\rm ii} \end{array} \right],\;
R_1 = \left[\begin{array}{cc} R  &  -G^*S_{\rm e}  \\ G_{\rm e} & S_{\rm ee}-I \end{array}\right],\;
M_1 =  \left[\begin{array}{c}  -G^*S_{\rm i} \\  S_{\rm ei} \end{array}\right],\\
G_1 &= \left[\begin{array}{cc} G_{\rm i} & S_{\rm ie}\end{array}\right],\; Z_1 = \left[\begin{array}{cc} Z & -C^*S_{\rm e}\end{array}\right],\; X_1 =\left[\begin{array}{c} X  \\ C_{\rm e}\end{array}\right].
\end{align*}
We therefore obtain the feedback reduction 
\[
\mathcal{F}\mathbf{G} \left( k\right) =G_{\rm ee}\left( k\right) \equiv \lbrack I,ka^{\ast
}]\left( g/g_{\rm ii}\right) \left[ 
\begin{array}{c}
I \\ 
ka
\end{array}
\right] .
\]

Conversely, the adiabatic elimination corresponds to 
\[
\mathcal{A}\mathbf{G}\left( k\right) =g/g_{ff}.
\]
In this way we see that the essential action is a Schur complementation of
the object $g$ either with respect to labels of the fast oscillators of the
system, or the labels of the internal fields. To this end, we can now establish the main technical result of this paper.

\begin{theorem}
\label{th:AE-IF-commute} Let $\mathbf{G}\left( k\right)$ and $\mathcal{F}\mathbf{G}\left( k\right)$ correspond to strictly Hurwitz stable open quantum systems (i.e., the $A$ matrix of each system is strictly Hurwitz stable), and suppose that $S_{\rm ii}+C_{\rm i}A^{-1}C^*S_{\rm i}-I$ and $S_{\rm ii}-I$ are invertible. Then in the notation established above we have 
\[
\mathcal{AF}\mathbf{G}\left( k\right) =\mathcal{FA}\mathbf{G} \left( k\right) .
\]
\end{theorem}

The proof of the above theorem is given in the Appendix. Thus we establish that  that  if $\mathbf{G}\left( k\right)$ and $\mathcal{F}\mathbf{G}\left( k\right)$ are strictly Hurwitz stable systems, and $S_{\rm ii}+C_{\rm i}A^{-1}C^*S_{\rm i}-I$ and $S_{\rm ii}-I$ are invertible, the operation of adiabatic elimination of the oscillators in the network indeed commutes with the operation of taking the instantaneous feedback limit. For the systems in series example of section \ref{sec:sys-series} it can be seen that the strictly Hurwitz stable property holds when $A_1$ and $A_2$ are strictly Hurwitz stable, while for the beam splitter with an in-loop device of section \ref{sec:sys-in-loop} it holds when $|\alpha|<1$.

The requirement that $\mathbf{G}\left( k\right)$ and $\mathcal{F}\mathbf{G}\left( k\right)$ be strictly Hurwitz is due to Remark \ref{rem:strict_Hurwitz}. Note that the strict Hurwitz condition is not however necessary and that the limits may more generally commute whenever the kernel property of $Y$ in Theorem \ref{thm: adiabatic elimination} holds. 
\section{Conclusion}
In this paper we have studied the question of commutativity of adiabatic elimination of oscillatory components and the operation of taking the instantaneous feedback limit in a quantum network with Markovian components. Provided some mild conditions are satisfied, we answer the question in the affirmative by showing that adiabatic elimination can be viewed as a Schur complementation operation, thus putting it on the same footing as the instantaneous feedback limit, and subsequently proving the commutativity of successive Schur complementation. This result is important from a practical point of view because in practice it is much easier to obtain a simplified description of a quantum network by first obtaining simplified component models and then using them to obtain a description of the network rather than the converse operation of first forming the (possibly large) network and applying adiabatic elimination at the network level. Since we have shown that the order in which adiabatic elimination and the instantaneous feedback limit is taken is inconsequential, this  justifies  employing the former order of operations which is free of any concerns regarding the uniqueness of the resulting simplified network model in which the fast oscillatory components have been eliminated.

\section{Appendix}

\subsection{Proof of Theorem \ref{thm: adiabatic elimination}}

Let us set 
\[
M_{N}=\mathfrak{\hat{h}}\otimes \mathrm{span}\left\{ |\mathbf{n}\rangle :\sum
n_{j}=N\right\}, 
\]
for $N=0,1,2,\ldots$. In particular, we have the direct sum of orthogonal subspaces $\mathfrak{\hat{h}}%
\otimes \mathfrak{h}_{\mathrm{osc}}=\oplus _{N\geq 0}M_{N}$. Let $P_{s}$ be
orthogonal projection onto the ``slow subspace'' $M_0=\mathfrak{h}_s=\mathfrak{\hat{h}}\otimes 
\mathbb{C}|0\rangle _{\mathrm{osc}}$ and let $P_{f}=I-P_{s}.$ Recall the hypothesis
that ${\rm ker}(Y)=\mathfrak{h}_s$. We first have the following:

\begin{lemma}
Under the hypothesis ${\rm ker}(Y)=\mathfrak{h}_s$, the subspaces $M_{N}$ are stable under $Y_{N}=Y|_{M_{N}}$, and we
have 
\[
\left( P_{f}YP_{f}\right) ^{-1}=\oplus _{N\geq 1}Y_{N}^{-1}.
\]
Moreover, let $|\delta _{j}\rangle $ be the state where the $j$-th mode is in
the first excited state and all others are in the vacuum, then $\left(
Y_{1}\right) ^{-1}\sum_{j}\phi _{j}\otimes |\delta _{j}\rangle
=\sum_{jl}\left( A^{-1}\right) _{jl}\phi _{l}\otimes |\delta _{j}\rangle $.
\end{lemma}

\begin{proof}
Stability and invertibility of $Y_{N}$ on $M_{N}$ follows directly from the
specific form of $Y_{N}$ and the fact that $Y$ has kernel space $M_{0}$. The
relation $\left( P_{f}YP_{f}\right) ^{-1}=\oplus _{N\geq 1}Y_{N}^{-1}$
follows from the direct sum decomposition.

The remaining identity is easily checked from $Y\sum_{j}\phi _{j}\otimes
|\delta _{j}\rangle =\sum_{jl}A_{jl}\phi _{l}\otimes |\delta _{j}\rangle $
and setting this equal to $\sum_{j}\tilde{\phi}_{j}\otimes |\delta
_{j}\rangle $ we deduce that $\phi _{l}=\left( A^{-1}\right) _{lj}\tilde{\phi%
}_{j}$.
\end{proof}

\begin{corollary}
${\rm ker}(Y^*)=M_0$. 
\end{corollary}
\begin{proof}
By the preceding lemma we have that $\mathfrak{h}_{\rm f}=P_f \hat{\mathfrak{h}} \otimes  \mathfrak{h}_{\rm osc}$ is stable under $Y$. Therefore, for any $\phi \in M_0$ and $\psi \in \hat{\mathfrak{h}} \otimes  \mathfrak{h}_{\rm osc}$ we have
that $\langle \phi, Y\psi \rangle =\langle \phi, YP_f \psi \rangle =0$. It follows that $\langle Y^*\phi, \psi \rangle = \langle \phi, Y \psi \rangle=0$ for all $\psi \in \hat{\mathfrak{h}} \otimes  \mathfrak{h}_{\rm osc}$, thus $Y^* \phi=0$ for any $\phi \in M_0$ and we conclude that $ M_0 \subseteq {\rm ker}(Y^*)$. We now need to show the converse that ${\rm ker}(Y^*) \subseteq M_0$ and we will do this by contradiction. To do this end, suppose that $\exists\varphi \in P_f \hat{\mathfrak{h}} \otimes  \mathfrak{h}_{\rm osc}$ with $\varphi \neq 0$ such that $\langle Y^*\varphi, \psi \rangle =0$ for all $\psi \in \hat{\mathfrak{h}} \otimes  \mathfrak{h}_{\rm osc}$. It follows that $\langle \varphi, Y\psi \rangle =0$ and therefore $\langle \varphi, YP_f \psi \rangle =0$ for all $\psi \in \hat{\mathfrak{h}} \otimes  \mathfrak{h}_{\rm osc}$. But since $\mathfrak{h}_{\rm f}$ is stable under $Y$ and $Y|_{\mathfrak{h}_f}$ is invertible, it follows that $\varphi \in \mathfrak{h}_{\rm s}$. But this contradicts the hypothesis that $\varphi$ is a non-zero element of $\mathfrak{h}_f$ and therefore we conclude that ${\rm ker}(Y^*) \subseteq M_0$. This concludes the proof.
\end{proof}

We now state a sufficient condition for ${\rm ker}(Y)=M_0={\rm ker}(Y^*)$. Let us first recall the following definition
\begin{definition}
A bounded Hilbert space operator $A$ is strictly Hurwitz stable if 
\[
{\rm Re}\langle \psi |A\psi \rangle <0\text{, \ for all }\psi \neq 0.
\]
\end{definition}

\begin{lemma}
\label{lem:Hurwitz} Let $A_{jl}\in \mathcal{B}(\mathfrak{\hat{h}})$ such
that $A=\left[ A_{jl}\right] \in \mathcal{B}(\mathfrak{\hat{h}}\otimes 
\mathbb{C}^{m})$ is strictly Hurwitz stable. Then the operator 
\begin{equation}
Y=\sum_{jl}A_{jl}\otimes a_{j}^{\ast }a_{l}  \label{eq:Y}
\end{equation}
on $\mathfrak{\hat{h}}\otimes \mathfrak{h}_{\text{osc}}$ has kernel consisting of
vectors of the form $\phi \otimes |0\rangle _{\mathrm{osc}}$, where $\phi \in 
\mathfrak{\hat{h}}$.
\end{lemma}

\begin{proof}
We see that for $\psi \in \mathfrak{\hat{h}}\otimes \mathfrak{h}_{\text{osc}}$%
\[
\langle \psi |Y\psi \rangle =\sum_{jl}\langle \psi |\left( I\otimes
a_{j}\right) ^{\ast }\left( A_{jl}\otimes I\right) \left( I\otimes
a_{l}\right) \psi \rangle =\sum_{jl}\langle \psi _{j}|\,A_{jl}\otimes
I\,\psi _{l}\rangle
\]
where $\psi _{j}=\left( I\otimes b_{j}\right) \psi $. We may decompose $\psi
_{j}\equiv \sum_{\mathbf{n}}\psi _{j}\left( \mathbf{n}\right) \otimes |%
\mathbf{n}\rangle $, where $|\mathbf{n}\rangle $ is the orthonormal basis of
number states for the oscillators and $\psi _{j}\left( \mathbf{n}\right) \in 
\mathfrak{\hat{h}}$. Then 
\[
\langle \psi |Y\psi \rangle =\sum_{\mathbf{n}}\sum_{jl}\langle \psi
_{j}\left( \mathbf{n}\right) |\,A_{jl}\,\psi _{l}\left( \mathbf{n}\right)
\rangle
\]
and, for each fixed $\mathbf{n}$, we have $\sum_{jl}\langle \psi _{j}\left( 
\mathbf{n}\right) |\,A_{jl}\,\psi _{l}\left( \mathbf{n}\right) \rangle \leq
0 $ with equality if and only if the $\psi _{j}\left( \mathbf{n}\right) =0$
since $\left[ A_{jl}\right] $ is assumed to be strictly Hurwitz. In
particular, if we assume that $\psi $ is in the kernel of $Y$ then we deduce
that $\psi _{j}\left( \mathbf{n}\right) =0$ for each $\mathbf{n}$ and $%
j=1,\cdots ,m$. It follows that $\psi _{j}=\left( I\otimes b_{j}\right) \psi
=0$ for each $j=1,\cdots ,m$, and this implies that $\psi \equiv \phi
\otimes |0\rangle_{\rm osc} $ for some $\phi \in \mathfrak{\hat{h}}$ as required.
\end{proof}

Note, however, that as we shall see below, for Theorem \ref{thm: adiabatic elimination}
to hold it is enough that ${\rm ker}(Y)=M_0$.

\begin{lemma}
\label{lem:unitary coefficients} The operator $\hat{S}$ is unitary and $\hat{%
K}+\hat{K}^{\ast }+\hat{L}^{\ast }\hat{L}=0$.
\end{lemma}

\begin{proof}
We first show that $I+CA^{-1}C^{\ast }$ is invertible. Suppose that $u\in 
\mathrm{ker}\left( I+CA^{-1}C^{\ast }\right) $%
\begin{eqnarray*}
u &=&-CA^{-1}C^{\ast }u\Rightarrow C^{\ast }u=-C^{\ast }CA^{-1}C^{\ast
}u\Rightarrow \left( I+C^{\ast }CA^{-1}\right) C^{\ast }u=0 \\
&\Rightarrow &\left( A+C^{\ast }C\right) A^{-1}C^{\ast }u=0\Rightarrow
-A^{\ast }A^{-1}C^{\ast }u=0\Rightarrow C^{\ast }u=0
\end{eqnarray*}
so substituting $C^{\ast }u=0$\ into $u=-CA^{-1}C^{\ast }u$ we see that $u=0$%
, therefore $\mathrm{ker}\,\hat S =0$. As $S$ is unitary, we have 
\begin{eqnarray*}
\hat{S}\hat{S}^{\ast } &=&\left( I+C^\ast A^{-1}C\right) \left( I+CA^{\ast
-1}C^{\ast }\right) \\
&=&I+CA^{-1}\left( A+A^{\ast }+C^{\ast }C\right) A^{\ast -1}C^{\ast }=I
\end{eqnarray*}
using the first of identities (\ref{identities}). Similarly $\hat{S}^{\ast }%
\hat{S}=I$.

Likewise we use (\ref{identities}) to show that 
\begin{eqnarray*}
\hat{K}+\hat{K}^{\ast }+\hat{L}^{\ast }\hat{L} &=&R-XA^{-1}Z+R^{\ast
}-Z^{\ast }A^{\ast -1}X^{\ast } \\
&&+\left( G^{\ast }-Z^{\ast }A^{\ast -1}C^{\ast }\right) \left(
G-CA^{-1}Z\right) \\
&=&-\left( X-G^{\ast }C\right) A^{-1}Z-Z^{\ast }A^{\ast -1}\left( X^{\ast
}-C^{\ast }G-C^{\ast }CA^{-1}Z\right) \\
&=&Z^{\ast }A^{-1}Z+Z^{\ast }A^{\ast -1}\left( Z+C^{\ast }CA^{-1}Z\right) \\
&=&Z^{\ast }A^{\ast -1}(A+A^{\ast }+C^{\ast }C)A^{-1}Z \\
&=&0.
\end{eqnarray*}
\end{proof}

Using the above results, we can now proceed to complete the proof  of Theorem \ref{thm: adiabatic elimination}.
Let us first recall the main results from Bouten, van Handel, and Silberfarb \cite{BvHS}. Let $V\left( t,k\right)
=U\left( t,k\right) ^{\ast }$, then $V$ satisfies the left QSDE (using a summation  convention)
\[
dV(t,k)=V\left( t,k\right) \left\{ 
\alpha \left( k\right) \otimes dt+\beta _{l}\left( k\right) \otimes dB_{l}(t)
+\gamma _{j}\otimes dB_{j}(t)^{\ast }+(\varepsilon _{jl}-\delta
_{jl})\otimes d\Lambda _{jl}(t)
\right\} ,
\]
where $\alpha \left( k\right) =k^{2}\alpha _{2}+k\alpha _{1}+\alpha
_{0}=K\left( k\right) ^{\ast }$, $\beta _{j}\left( k\right) =k\beta
_{1,j}+\beta _{0,j}=L_{j}\left( k\right) ^{\ast }$, $\gamma _{j}\left(
k\right) =-S_{lj}^{\ast }L_{l}$ and $\varepsilon _{jl}=S_{lj}^{\ast }$.
Their results are stated for the left QSDE for the reason that it is easier
to formulate the conditions for unbounded coefficients this way, however,
the treatment is of course equivalent.

We note that $\alpha _{2}$ is then $Y^{\ast }$, with kernel space $M_{0}$,
and we denote its Moore-Penrose inverse by $\tilde{\alpha}_{2}$  (note that this Moore-Penrose inverse exists
since $Y^*$ has the same form and properties as $Y$). The pre-limit coefficients satisfy Assumption 1 in the paper of Bouten, van Handel, and Silberfarb\cite{BvHS} by
construction. Assumption 2 of that work correspond to the identities $\alpha
_{2}\tilde{\alpha}_{2}=\tilde{\alpha}_{2}\alpha _{2}=P_{f}$, $\alpha
_{2}P_{s}=0$, $\beta _{1,i}^{\ast }P_{s}=0$ and $P_{s}\alpha _{1}P_{s}=0$:
the last three are automatic since $P_{s}$ projects onto the ground state of
the oscillator and in each case we encounter $a_{j}|0\rangle _{\mathrm{osc}%
}=0$. The limit coefficients in Assumption 3 of \cite{BvHS} are then 
\begin{eqnarray*}
\hat{\alpha} &=&P_{s}\left( \alpha _{0}-\alpha _{1}\tilde{\alpha}_{2}\alpha
_{1}\right) P_{s}=\left( R^{\ast }-Z^{\ast }A^{\ast -1}X^{\ast }\right)
\otimes |0\rangle \langle 0|_{\mathrm{osc}}\equiv \hat{K}^{\ast }\otimes
|0\rangle \langle 0|_{\mathrm{osc}}, \\
\hat{\beta} &=&P_{s}\left( \beta _{0}-\alpha _{1}\tilde{\alpha}_{2}\beta
_{1}\right) P_{s}=\left( G^{\ast }-Z^{\ast }A^{\ast -1}C^{\ast }\right)
\otimes |0\rangle \langle 0|_{\mathrm{osc}}\equiv \hat{L}^{\ast }\otimes
|0\rangle \langle 0|_{\mathrm{osc}}, \\
\hat{\varepsilon} &=&P_{s}\varepsilon \left( I+\beta _{1}^{\ast }\tilde{%
\alpha}_{2}\beta _{1}\right) P_{s}=S^{\ast }\left( I+C^{\ast }A^{\ast
-1}C^{\ast }\right) \otimes |0\rangle \langle 0|_{\mathrm{osc}}\equiv \hat{S}%
^{\ast }\otimes |0\rangle \langle 0|_{\mathrm{osc}}, \\
\hat{\gamma} &\equiv &-\hat{\varepsilon}\hat{\beta}^{\ast }\equiv -\hat{S}%
^{\ast }\hat{L}\otimes |0\rangle \langle 0|_{\mathrm{osc}},
\end{eqnarray*}
with $(\hat{S},\hat{L},\hat{K})$ as given in the statement of
Theorem \ref{thm: adiabatic elimination}. These coefficients\ evidently satisfy the requirements of
Assumption 3, namely to generate a unitary adapted Hudson-Parthasarathy
equation on a common invariant domain in $M_{0}$, as was established in
Lemma \ref{lem:unitary coefficients}.
$\blacksquare$

\subsection{Conditions for the Schur complements in Lemma \ref{lem:scrule} to be well-defined}

\begin{lemma}
\label{new}
If 
\begin{eqnarray}
 \ker \left[\begin{array}{cc}
	M_{B,B} & M_{B,C}\\
	M_{C,B} & M_{C,C}
      \end{array}\right]
 &\subseteq& 
\ker\,[M_{A,B}\ M_{A,C}] \label{eq:condition1}\\
%---------------------------------
\mathrm{im} \left[\begin{array}{c}
                              M_{B,A}\\
			      M_{C,A}
                             \end{array}\right]
&\subseteq& 
{\rm im}\,\left[\begin{array}{cc}
	M_{B,B} & M_{B,C}\\
	M_{C,B} & M_{C,C}
      \end{array}\right]\label{eq:condition2}\\
%---------------------------------
\ker M_{C,C} &\subseteq& \ker M_{B,C} \label{eq:condition3}\\
\mathrm{im}\, M_{C,B} &\subseteq& \mathrm{im}\, M_{C,C} \label{eq:condition4} \\
\ker M_{B,B} &\subseteq& \ker M_{C,B} \label{eq:condition5}\\
\mathrm{im}\, M_{B,C} &\subseteq& \mathrm{im}\, M_{B,B}, \label{eq:condition6}
\end{eqnarray}
then the Schur complements 
%$ M/M_{B \cup C,B \cup C},\ M_{B \cup C, B \cup C}/ M_{B,B},$ and $ M_{B \cup C, B
%\cup C}/ M_{C,C}$
%are well-defined
$(M/M_{C,C})/(M/M_{C,C})_{B,B} $ and $(M/M_{B,B})/((M/M_{B,B}))_{C,C}$ are all well-defined.
\end{lemma}

\begin{proof}
Collecting the Schur complements used in the proof of Lemma \ref{lem:scrule} (successive
complementation rule), we see that we have to show that
\begin{eqnarray*}
\ M/M_{C,C},\ M/M_{B,B},\
(M/M_{C,C})/(M/M_{C,C})_{B,B}\\
M_{A \cup C, A \cup C}/M_{C,C},\ M_{A \cup C, B \cup C}/M_{C,C},\ M_{B \cup C, A
\cup C}/M_{C,C}
\end{eqnarray*}
exist. To proceed, first note that, by Lemma \ref{lm:GS-exist}, (\ref{eq:condition1})-(\ref{eq:condition6}) imply that \[ M/M_{B \cup C,B \cup C},\ M_{B \cup C, B \cup C}/ M_{B,B},\ M_{B
\cup C, B
\cup C}/ M_{C,C} \] are well-defined.
From $\ker M_{B,B} \subseteq \ker M_{C,B}$, we see that $M_{B,B}x = 0 \Rightarrow
M_{C,B}x = 0$. This combined with condition (\ref{eq:condition1}) shows
that $M_{B,B}x = 0 \Rightarrow \left[\begin{array}{c} M_{B,B} \\
M_{C,B} \end{array}\right]x = 0 \Rightarrow M_{A,B}x = 0$. Thus 
(\ref{eq:1}), given below, holds. Now,  (\ref{eq:condition2}) implies that $\forall\ x\
\exists\ y,z$ such that 
\begin{equation}
\left[\begin{array}{c} M_{B,A} \\
M_{C,A} \end{array}\right]x = \left[\begin{array}{c} M_{B,B}y + M_{B,C}z\\
M_{C,B}y + M_{C,C}z \end{array}\right]. \label{eq:7}
\end{equation}
But conditions  (\ref{eq:condition4}) and (\ref{eq:condition6}) imply that
$\exists\ w,v$ such that $M_{C,B}y = M_{C,C}v$ and $M_{B,C}z = M_{B,B}w$. This
together with (\ref{eq:7}) shows that $\mathrm{im}\,M_{B,A} \subseteq
\mathrm{im}\, M_{B,B}$ and $\mathrm{im}\, M_{C,A} \subseteq
\mathrm{im}\, M_{C,C}$. Combining this with (\ref{eq:condition6}) gives us
 (\ref{eq:2}), given below. By analogous arguments we can also establish (\ref{eq:3})
and (\ref{eq:4}). 
\begin{eqnarray}
 \ker M_{B,B} &\subseteq& \ker \left[\begin{array}{c}
                              M_{A,B}\\
			      M_{C,B}
                             \end{array}\right]\label{eq:1}\\
\mathrm{im}\,[\ M_{B,A}\ M_{B,C}] &\subseteq& \mathrm{im}\,M_{B,B} \label{eq:2}\\
\ker M_{C,C} &\subseteq& \ker \left[\begin{array}{c}
                              M_{A,C}\\
			      M_{B,C}
                             \end{array}\right] \label{eq:3}\\
\mathrm{im}\,[M_{C,A}\ M_{C,B}] &\subseteq& \mathrm{im}\,M_{C,C} \label{eq:4}
\end{eqnarray}
From (\ref{eq:1}) to (\ref{eq:4}) it follows directly that $M/M_{B,B}$ and $M/M_{C,C}$ exist.
Existence of $M_{A \cup C, A \cup C}/M_{C,C}$,  $M_{A \cup C, B \cup
C}/M_{C,C}$ and $M_{B \cup C, A \cup
C}/M_{C,C}$ follows immediately from  (\ref{eq:3}) and (\ref{eq:4}).

Now we show that $(M/M_{C,C})/(M/M_{C,C})_{B,B}$ exists. We require
\begin{eqnarray*}
 \ker (M_{B,B} - M_{B,C} M_{C,C}^- M_{C,B}) &\subseteq& \ker (M_{A,B} - M_{A,C}
M_{C,C}^-M_{C,B}) \\
 \mathrm{im}(M_{B,A} - M_{B,C} M_{C,C}^- M_{C,A}) &\subseteq&
\mathrm{im}(M_{B,B} - M_{B,C} M_{C,C}^- M_{C,B})
\end{eqnarray*}
Let $v \in \mathrm{im} (M_{B,A} - M_{B,C} M_{C,C}^- M_{C,A})$ be
\begin{eqnarray*}
 v &=& (M_{B,A} - M_{B,C} M_{C,C}^- M_{C,A})x\\
   &=& [\begin{array}{cc} 1 & -M_{B,C} M_{C,C}^-\end{array}] \left[\begin{array}{c}
                              M_{B,A}x\\
			      M_{C,A}x
                             \end{array}\right],
\end{eqnarray*}
for some vector $x$. Using (\ref{eq:condition2}) and noting that $M_{B,C}M_{C,C}^- M_{C,C} =
M_{B,C}$ (due to (\ref{eq:condition3}) and  Lemma \ref{lm:RL-inverse}) leads to
\begin{eqnarray*}
 v &=& [\begin{array}{cc} 1 & -M_{B,C} M_{C,C}^-\end{array}] \left[\begin{array}{c}
	M_{B,B}y + M_{B,C}z\\
	M_{C,B}y + M_{C,C}z
      \end{array}\right]\\
  &=& (M_{B,B} - M_{B,C} M_{C,C}^- M_{C,B})x
\end{eqnarray*}
which shows the required image space inclusion.

To show that $ \ker (M_{B,B} - M_{B,C} M_{C,C}^- M_{C,B}) \subseteq \ker
(M_{A,B} - M_{A,C} M_{C,C}^-M_{C,B})$ holds we choose some $x \in \ker (M_{B,B}
- M_{B,C} M_{C,C}^- M_{C,B})$ and see that
\[
 M_{B,B}x - M_{B,C} M_{C,C}^-M_{C,B}x = [\begin{array}{cc} M_{B,B} & M_{B,C} \end{array}]
\left[\begin{array}{c}
	x\\
	-M_{C,C}^- M_{C,B}x
\end{array}\right] = 0
\]
which implies that $
\left[\begin{array}{c}
	x\\
	-M_{C,C}^- M_{C,B}x
\end{array}\right] \in 
\ker [\begin{array}{cc} M_{B,B} & M_{B,C} \end{array}]$. However, $M_{C,B}x - M_{C,C} M_{C,C}^- M_{C,B}x = 0$
since $M_{C,C} M_{C,C}^- M_{C,B} = M_{C,B}$ (by (\ref{eq:4})  and Lemma \ref{lm:RL-inverse}).
It then follows that
\[
\left[\begin{array}{c}
	x\\
	-M_{C,C}^- M_{C,B}x
\end{array}\right] \in 
\ker \left[\begin{array}{cc}
	M_{B,B} & M_{B,C}\\
	M_{C,B} & M_{C,C}
      \end{array}\right]
\]
which  by (\ref{eq:condition1}) implies $
\left[\begin{array}{c}
	x\\
	-M_{C,C}^- M_{C,B}x
\end{array}\right] \in \ker\, [M_{A,B}\ M_{A,C}]$,
and therefore we deduce $(M_{A,B} - M_{A,C}
M_{C,C}^-M_{C,B})x =0$ which proves the required kernel space inclusion.
\end{proof}

We remark that the conditions (\ref{eq:condition1})-(\ref{eq:condition6}) are not necessary, as is clear from Horn and Zhang \cite{Zhang}, page 42. 

\subsection{Proof of Theorem \ref{th:AE-IF-commute}}
The model may initially be described by the set of coefficients $g$ over the
space 
\[
\mathfrak{H}=(\mathfrak{\hat{h}}\oplus \mathfrak{\hat{h}})\otimes \left( \mathbb{%
C}\oplus \mathfrak{K}_{\rm e}\oplus \mathfrak{K}_{\rm i}\right)
\]
and we decompose this as 
\[
\mathfrak{H}=\mathfrak{H}_{1}\oplus \mathfrak{H}_{2}\oplus \mathfrak{H}%
_{3}\oplus \mathfrak{H}_{4}
\]
where 
\[
\begin{array}{lcc}
\mathfrak{H}_{1}=\mathfrak{\hat{h}}\otimes \left( \mathbb{C}\oplus %
\mathfrak{K}_{\rm e}\right) , & \mathrm{Slow} & \mathrm{External} \\ 
\mathfrak{H}_{2}=\mathfrak{\hat{h}}\otimes \mathfrak{K}_{\rm i}, & \mathrm{Slow}
& \mathrm{Internal} \\ 
\mathfrak{H}_{3}=\mathfrak{\hat{h}}\otimes \left( \mathbb{C}\oplus \mathfrak{K}%
_{\rm e}\right) , & \mathrm{Fast} & \mathrm{External} \\ 
\mathfrak{H}_{4}=\mathfrak{\hat{h}}\otimes \mathfrak{K}_{\rm i}, & \mathrm{Fast} & 
\mathrm{Internal}
\end{array}.
\]
With respect to this decomposition, we decompose $g$ into sub-blocks as 
\begin{equation}
g=\left[ 
\begin{array}{cccc}
g_{11} & g_{12} & g_{13} & g_{14} \\ 
g_{21} & g_{22} & g_{23} & g_{24} \\ 
g_{31} & g_{32} & g_{33} & g_{34} \\ 
g_{41} & g_{42} & g_{43} & g_{44}
\end{array}
\right] \equiv \left[ \begin{array}{cccc} R_1 & M_1 & X_1 & 0\\
G_1 & S_{\rm ii} -I & C_{\rm i}  & 0\\
Z_1& -C^*S_{\rm i} & A & 0\\
0 & 0 & 0 & 0  
 \end{array}\right], \label{eqn: g decomp}
\end{equation}
The set of labels $\mathfrak{I}=\left\{ 1,2,3,4\right\} $ can be split up
into the slow labels $\left\{ 1,2\right\} $ and the fast labels $F=\left\{
3,4\right\} =S^{\prime }$ as well as the external labels $E=\left\{
1,3\right\} $ and the internal labels $I=\left\{ 2,4\right\} =E^{\prime }$.

To proceed, we must first establish that the generalized Schur complement is
well-defined. Here we are ultimately retaining the ``slow external'' degrees
of freedom (index 1) and eliminating the index sets $\left\{ 2,3,4
\right\}$. To this end, We need to check that  conditions (\ref{eq:condition1})-(\ref{eq:condition6}) are all satisfied. We begin with (\ref{eq:condition1}). 

Let $(x,y,z)^T$ be an element of ${\rm ker}\left[ \begin{array}{ccc} S_{\rm ii}-I & C_{\rm i} & 0\\
-C^*S_{\rm i} & A & 0 \\ 0 & 0 & 0 \end{array}\right]$. Then $x,y$ satisfies $(S_{\rm ii}-I)x +C_{\rm i}y=0$ and  $-C^*S_{\rm i}x +Ay=0$, while $z$ is arbitrary. Therefore, we have $y= A^{-1}C^*S_{\rm i}x$ and $(S_{\rm ii}+C_{\rm i}A^{-1}C^*S_{\rm i}-I)x=0$. Since $S_{\rm ii}+C_{\rm i}A^{-1}C^*S_{\rm i}-I$ is invertible by hypothesis, we find that   $x=0$. It then follows that also $y=0$, and we conclude that ${\rm ker}\left[ \begin{array}{ccc} S_{\rm ii}-I  & C_{\rm i} & 0\\
-C^*S_{\rm i} & A & 0 \\ 0 & 0 & 0 \end{array}\right]$ consists of vectors of the form $(0,0,z)^T$. Clearly such vectors lie in the kernel of $[\begin{array}{ccc} M_1 & X_1 & 0 \end{array}]$ and we conclude that ${\rm ker}\left[ \begin{array}{ccc} S_{\rm ii}-I & C_{\rm i} & 0\\
-C^*S_{\rm i} & A & 0 \\ 0 & 0 & 0 \end{array}\right] \subseteq {\rm ker}[\begin{array}{ccc} M_1 & X_1 & 0 \end{array}]$.

Next, we check if for every given vector $x$ there exist vectors $y$ and $z$ such that we have the equality
\begin{equation}
\left[\begin{array}{c} G_1x \\ Z_1x  \end{array} \right]=\left[\begin{array}{cc} S_{\rm ii}-I & C_{\rm i}  \\
-C^*S_{\rm i} & A  \end{array} \right] \left[ \begin{array}{c} y \\  z \end{array}\right]. \label{eq:aux-im-incl}
\end{equation}
In particular, this will be satisfied if the matrix $\left[\begin{array}{cc} S_{\rm ii}-I & C_{\rm i}  \\
-C^*S_{\rm i} & A  \end{array} \right]$ is invertible. However,  since $S_{\rm ii} -I + C_{\rm i}A^{-1}C^*S_{\rm i}$ and $A$ are invertible we see that this simply follows from the Banachiewicz formula. Therefore, for any vector $x$ there indeed exist vectors $y$ and $z$ such that (\ref{eq:aux-im-incl}) holds and we conclude that ${\rm im}\left[ \begin{array}{cc} G_1 \\ Z_1 \\ 0 \end{array} \right] \subseteq {\rm im}\left[\begin{array}{ccc} S_{\rm ii} & C_{\rm i}  & 0 \\
-C^*S_{\rm i} & A  & 0 \\ 0 & 0 & 0 \end{array} \right]$. Moreover, from the fact that $A$ is invertible we also  get 
\begin{eqnarray*}
{\rm ker}\left[ \begin{array}{cc} A & 0 \\ 0 & 0 \end{array}\right] &\subseteq& {\rm ker}\,[\begin{array}{cc} C_{\rm i} & 0 \end{array}],\\
{\rm im}\left[ \begin{array}{c} -C^*S_{\rm i} \\ 0 \end{array}\right] &\subseteq& {\rm im}\left[\begin{array}{cc} A  & 0 \\ 0 & 0 \end{array}\right],
\end{eqnarray*}
while from the invertibility of $S_{\rm ii}-I$ we automatically have 
\begin{eqnarray*}
{\rm ker}(S_{\rm ii}-I) &\subseteq& {\rm ker}\left[\begin{array}{cc}  -C^*S_{\rm i} \\ 0 \end{array}\right],\\
{\rm im}\,[\begin{array}{cc} C_{\rm i} & 0 \end{array}] &\subseteq& {\rm im}(S_{\rm ii}-I).
\end{eqnarray*}
Therefore conditions  (\ref{eq:condition1})-(\ref{eq:condition6}) are satisfied, and the theorem now follows from Lemmata \ref{new} and \ref{lem:scrule}. 
 
\textbf{Acknowledgement}
J. Gough and H. Nurdin acknowledge the support of EPSRC research grant EP/H016708/1
``Quantum Control: Feedback Mediated by Channels in Non-classical States". H. Nurdin also acknowledges the support of the Australian Research Council. The authors thank an anonymous reviewer for bringing their attention to survey paper of Ouellette\cite{DVO81}.


\begin{thebibliography}{99}
\bibitem{Gardiner}  C. Gardiner and P. Zoller, Quantum Noise: A Handbook of
Markovian and Non-Markovian Quantum Stochastic Methods with Applications to
Quantum Optics, 2nd ed., ser. Springer Series in Synergetics. Springer,
(2000).

\bibitem{HP}  R. L. Hudson and K. R. Parthasarathy, \emph{Quantum Ito's
formula and stochastic evolutions,} Commun. Math. Phys. \textbf{93}, 301-323
(1984)

\bibitem{partha}  K. Parthasarathy, An Introduction to Quantum Stochastic
Calculus. Berlin: Birkhauser, (1992)

\bibitem{GoughJamesCMP09}  J. Gough, M.~R. James, \emph{Quantum Feedback
Networks: Hamiltonian Formulation}, Commun. Math. Phys., Volume \textbf{287}%
, 1109-1132, Number 3 / May, (2009)

\bibitem{GoughJamesIEEE09}  J. Gough, M.~R. James, \emph{The Series Product
and Its Application to Quantum Feedforward and Feedback Networks}, IEEE
Trans. Automatic Control, \textbf{54}(11):2530-2544, (2009)

\bibitem{DVO81} D.V. Ouellette, \emph{Schur Complement and Statistics}, Linear Algebra and its applications,
36:187-295 (1981)

\bibitem{Zhang}  R. Horn and F. Zhang, \emph{Basic Properties of the Schur
Complement}, in The Schur Complement and Its Applications (Ed. F. Zhang),
Numerical Methods and Algorithms, Vol. \textbf{4}, Springer, (2005)

\bibitem{Ban37}  T. Banachiewicz, \emph{Zur Berechnung der Determination,
Wie auch der Inversen, und der darauf basierten Aufl\"{o}sung der Systeme
linearer Gleichungen}, Acta Astronomica, S\`{e}rie C, 3, 4167, (1937)

\bibitem{MS} G. Marsaglia, G.P.H. Styan, \emph{Rank conditions for generalized inverses of partitioned 
matrices}, Sankhy\={a} Ser. A 36: 437-442 (1974)

\bibitem{CHM} D. Carlson, E. Haynsworth, and T, Markham, \emph{A generalization of the Schur
complement by means of the Moore-Penrose inverse}, SIAM J. Appl. Math. 26: 169-175 (1974)

\bibitem{GoughvanHandel}  J.~E. Gough, R. van Handel, \emph{Singular
perturbation of quantum stochastic differential equations with coupling
through an oscillator mode}, J. Stat. Phys. \textbf{127}, 575-607 (2007)

\bibitem{BoutenSilberfarb}  L. Bouten, A. Silberfarb, \emph{Adiabatic
elimination in quantum stochastic models}, Commun. Math. Phys. \textbf{283},
491-505 (2008)

\bibitem{BvHS}  L. Bouten, R. van Handel, A. Silberfarb, \emph{Approximation
and limit theorems for quantum stochastic models with unbounded coefficients}%
, J. Funct. Analysis. \textbf{254}, 3123-3147 (2008)

\bibitem{GoughCMP}  J.~E. Gough, \emph{Quantum Flows as Markovian Limit of
Emission, Absorption and Scattering Interactions}, Commun. Math. Phys. Vol.
254, no.2, 489-512, March (2005)

\bibitem{GoughJMP}  J.~E. Gough, \emph{Quantum Stratonovich Stochastic
Calculus and the Quantum Wong-Zakai Theorem}, J. Math. Phys. 47, 113509,
(2006)

\end{thebibliography}
\end{document}